\documentclass[11pt]{article}

\title{Alea Iacta Est: Auctions, Persuasion, Interim Rules, and Dice}

\author{Shaddin Dughmi \\shaddin@usc.edu \and David Kempe \\
   dkempe@usc.edu \and Ruixin Qiang\\ rqiang@usc.edu}

\usepackage{fullpage}
\usepackage{diagbox}
\usepackage{bm}
\usepackage{amstext}

\usepackage{./definitions}
\usepackage{makecell}

\usepackage[numbers]{natbib}
\usepackage[ruled,noend,noline,boxed,linesnumbered]{algorithm2e}

\newcommand{\camera}[2]{#1}
\SetAlFnt{\small}
\SetAlCapFnt{\small}
\SetAlCapNameFnt{\small}
\SetAlCapHSkip{0pt}
\IncMargin{-\parindent}

\usepackage{enumitem}
\SetKwProg{Fn}{Function}{}{}
\SetKw{Continue}{continue}
\SetKw{Break}{break}
\SetKw{True}{true}
\SetKwInOut{Input}{Input}
\SetKwInOut{Output}{Output}

\newcommand{\vc}[1]{\ensuremath{\bm{#1}}\xspace}

\newcommand{\ConstructDice}{\ensuremath{\mbox{\textsc{Contruct Dice}}}\xspace}
\newcommand{\FindBarrierSet}{\ensuremath{\mbox{\textsc{Find Barrier Set}}}\xspace}
\newcommand{\Decrement}{\ensuremath{\mbox{\textsc{Decrement}}}\xspace}
\newcommand{\ConstructSymmetricDice}{\ensuremath{\mbox{\textsc{Contruct Symmetric Dice}}}\xspace}
\newcommand{\FindBarrierSetIID}{\ensuremath{\mbox{\textsc{Find Barrier Set IID}}}\xspace}
\newcommand{\DecrementSymmetricDice}{\ensuremath{\mbox{\textsc{Decrement Symmetric Dice}}}\xspace}

\newcommand{\DIESET}{\ensuremath{\mathcal{D}}\xspace}
\newcommand{\DIESETP}{\ensuremath{\mathcal{D}'}\xspace}
\newcommand{\DIE}[1][]{\ensuremath{%
\ifthenelse{\equal{#1}{}}{D}{D_{#1}}}\xspace}
\newcommand{\die}[2][]{\ensuremath{%
\ifthenelse{\equal{#1}{}}{D_{#2}}{D_{#1,#2}}}\xspace}
\newcommand{\DIEP}[1][]{\ensuremath{%
\ifthenelse{\equal{#1}{}}{D'}{D'_{#1}}}\xspace}
\newcommand{\diep}[2][]{\ensuremath{%
\ifthenelse{\equal{#1}{}}{D'_{#2}}{D'_{#1,#2}}}\xspace}
\newcommand{\INTE}{\ensuremath{\vc{x}}\xspace}
\newcommand{\INTEP}{\ensuremath{\vc{x'}}\xspace}
\newcommand{\INTEH}{\ensuremath{\vc{\hat{x}}}\xspace}
\newcommand{\interim}[2][]{\ensuremath{%
\ifthenelse{\equal{#1}{}}{x(#2)}{x_{#1}(#2)}}\xspace}
\newcommand{\interimp}[2][]{\ensuremath{%
\ifthenelse{\equal{#1}{}}{x'(#2)}{x'_{#1}(#2)}}\xspace}
\newcommand{\interimh}[2][]{\ensuremath{%
\ifthenelse{\equal{#1}{}}{\hat{x}(#2)}{\hat{x}_{#1}(#2)}}\xspace}
\newcommand{\SINTE}{\ensuremath{\vc{X}}\xspace}
\newcommand{\sinte}[3]{\ensuremath{x_{#1,#2,#3}}\xspace}
\newcommand{\TP}{\ensuremath{\vc{t}}\xspace}
\newcommand{\type}[1]{\ensuremath{t_{#1}}\xspace}

\newcommand{\PMF}[1][]{\ensuremath{%
\ifthenelse{\equal{#1}{}}{f}{f_{#1}}}\xspace}
\newcommand{\pmf}[2][]{\ensuremath{%
\ifthenelse{\equal{#1}{}}{f(#2)}{f_{#1}(#2)}}\xspace}
\newcommand{\PMFP}[1][]{\ensuremath{%
\ifthenelse{\equal{#1}{}}{f'}{f'_{#1}}}\xspace}
\newcommand{\pmfp}[2][]{\ensuremath{%
\ifthenelse{\equal{#1}{}}{f'(#2)}{f'_{#1}(#2)}}\xspace}
\newcommand{\VALF}[1][]{\ensuremath{%
\ifthenelse{\equal{#1}{}}{v}{v_{#1}}}\xspace}
\newcommand{\valf}[2][]{\ensuremath{%
\ifthenelse{\equal{#1}{}}{v_{#2}}{v_{#1,#2}}}\xspace}
\newcommand{\score}[1]{\ensuremath{\bar{v}_{#1}}\xspace}
\newcommand{\slack}{\ensuremath{\sigma}\xspace}
\newcommand{\NUMC}{\ensuremath{n}\xspace}

\newcommand{\NUMT}[1][]{\ensuremath{%
\ifthenelse{\equal{#1}{}}{m}{m_{#1}}}\xspace}
\newcommand{\NUMTP}[1][]{\ensuremath{%
\ifthenelse{\equal{#1}{}}{m'}{m'_{#1}}}\xspace}

\newcommand{\ALGO}[1][]{\ensuremath{%
\ifthenelse{\equal{#1}{}}{\mathcal{A}}{\mathcal{A}(#1)}}\xspace}

\newcommand{\cross}{\times}

\newcommand{\set}[1]{\left\{ #1 \right\}}

\newcommand{\intersect}{\cap}

\newcommand{\sm}{\setminus}

\renewcommand{\hat}{\widehat}
\renewcommand{\tilde}{\widetilde}
\renewcommand{\bar}{\overline}

\newcommand{\supp}{\text{supp}}

\DeclareMathOperator{\poly}{poly}
\DeclareMathOperator{\candidate}{\nu}

\def\min{\qopname\relax n{min}}
\def\max{\qopname\relax n{max}}

\newcommand{\NN}{\mathbb{N}}

\newcommand{\suf}{\bm{\beta}}
\newcommand{\pre}{\bm{\alpha}}

\def\sse{\subseteq}

\def\B{\mathcal{B}}
\def\C{\mathcal{C}}
\def\D{\mathcal{D}}
\def\E{\mathcal{E}}

\def\I{\mathcal{I}}

\def\M{\mathcal{M}}

\def\R{\mathcal{R}}

\begin{document}

\maketitle
\begin{abstract}
To select a subset of samples or ``winners'' from a population of
candidates, order sampling \cite{rosen1997asymptotic} and the $k$-unit
Myerson auction \cite{myerson} share a common scheme:
assign a (random) score to each candidate, then select the $k$
candidates with the highest scores.
We study a generalization of both order sampling and Myerson's allocation rule,
called \emph{winner-selecting dice}.
The setting for winner-selecting dice is similar to auctions with
feasibility constraints: candidates have random types drawn from
independent prior distributions,
and the winner set must be feasible subject to certain constraints.
Dice (distributions over scores) are assigned to each type,
and winners are selected to maximize the sum of the dice rolls,
subject to the feasibility constraints.
We examine the existence of winner-selecting dice that implement
prescribed probabilities of winning (i.e., an interim rule) for all types. 

Our first result shows that when the feasibility constraint is a matroid,
then for any feasible interim rule, there always exist
winner-selecting dice that implement it.
Unfortunately, our proof does not yield an efficient algorithm for
constructing the dice.
In the special case of a $1$-uniform matroid, i.e., only one winner
can be selected, we give an efficient algorithm that constructs
winner-selecting dice for any feasible interim rule.
Furthermore, when the types of the candidates are drawn in an 
i.i.d.~manner and the interim rule is symmetric across candidates,
unsurprisingly, an algorithm can efficiently construct symmetric dice that only
depend on the type but not the identity of the candidate.

One may ask whether we can extend our result to ``second-order'' interim rules,
which not only specify the winning probability of a type,
but also the winning probability conditioning on each other candidate's
type.
We show that our result does not extend, by exhibiting an
instance of \emph{Bayesian persuasion} whose optimal scheme is
equivalent to a second-order interim rule,
but which does not admit any dice-based implementation. 
\end{abstract}


\section{Introduction}
\label{sec:intro}


Economic design often features scenarios in which choices must be made
based on stochastic inputs. In auction design, bidders drawn from a
population interact with an auction mechanism, and the mechanism must
then choose winners and losers of the auction. In information
structure design, a principal observes information pertinent to the
various actions available to one or more decision makers, and must use
this information to recommend actions to the decision makers.
In such settings, an important framing device is the notion of
an \emph{interim rule} (also called a \emph{reduced form}) of an (ex-post) winner
selection rule, summarizing the probability for each candidate to be selected.

We focus on the simplest and most natural class of such decision
making scenarios, one which includes auctions and Bayesian persuasion
\citep{Kamenica11} as special cases. In a \emph{winner selection}
environment, there is a set of \emph{candidates} $\C$,
each equipped with a random attribute known as its \emph{type}.
A \emph{winner selection rule} is a
randomized function (or algorithm) which maps each profile of types,
one per candidate, to a choice of winning candidates,
subject to the requirement that the set of winners must belong to a
specified family $\I\subseteq 2^\C$ of feasible sets.
The winner selection rule is also referred to as an ex-post rule since
it specifies the winning probabilities conditioned on every realized
type profile.
In auctions, candidates correspond to bidders, and a winner selection
rule is an allocation rule of the auction.
In Bayesian persuasion, candidates correspond to actions available to
a decision maker, and a winner selection rule corresponds to a
persuasion scheme used by a principal to recommend one of the actions
to the decision maker.
We restrict our attention to winner selection scenarios in which the
types of different candidates are independently distributed.

We distinguish two classes of interim rules: \emph{first-order} and
\emph{second-order}. The former is the traditional notion from auction
theory, while the latter is the notion better suited for persuasion.
A \emph{first-order interim rule} specifies, for each candidate $i$ and
type $t$ of candidate $i$, the conditional probability of $i$ winning
given that his type is $t$. A \emph{second-order interim rule}
specifies more information:
for each pair of candidates $i,j$ and type $t$ of candidate $j$,
it specifies the conditional probability of $i$ 
winning given that $j$ has type $t$. 
First-order interim rules, when
combined with a payment rule, suffice for evaluating the welfare,
revenue, and incentive-compatibility of a single-item auction. For
Bayesian persuasion, second-order interim rules are needed for
evaluating the incentive constraints of a persuasion scheme. 

%

Our motivation for studying winner selection at this level of
generality stems from the success of Myerson's \cite{myerson}
famous and elegant characterization of revenue-optimal single-item
auctions when bidder type distributions are independent. In that
special case, Myerson showed that the optimal \emph{single-item} auction
features a particularly structured winner-selection rule: each type is
associated with a virtual value, and given a profile of reported
types, the rule selects the bidder with the highest (non-negative)
virtual value as the winner.

Interestingly, order sampling \cite{rosen1997asymptotic} works in a
similar way as Myerson's auction, in the special case when the feasible
sets are sets of $k$ or fewer candidates,
and each candidate has only one type.
Order sampling assigns each candidate a random score variable (a die).
The winners are the candidates with the $k$ highest score variables.
In the language of order sampling, virtual value functions
define a single-sided die for each type.
Unlike Myerson's auction, there is no notion of ``revenue'' or
``incentive constraints'' in order sampling.
The task
is simply to find dice that will induce a
prescribed first-order interim rule. 

As a generalization of order sampling and Myerson's virtual-value approach, a \emph{dice-based winner-selection rule} assigns each type a die, and selects the feasible set of winners maximizing the sum of the dice rolls. 
\emph{We explore the extent to which dice-based rules are applicable
  beyond single-parameter auctions or single-type environments, to
  winner selection with independent type distributions
  under more complex constraints.}
In particular, we examine whether dice-based winner-selection rules
exist for winner selection subject to matroid constraints and for
Bayesian persuasion.

\subsection{Our Results}


As mentioned previously, all of our results are restricted to settings
in which the candidates' type distributions are independent. 
It follows from Myerson's characterization that every first-order
interim rule \emph{corresponding to some optimal auction} admits a
dice-based implementation.%
\footnote{As we show in \camera{the appendix}{the full version}, there
  are interim rules which are not optimal for any auction.}
Our main result (in Section~\ref{sec:exist})
is an existential proof showing that \emph{every}
feasible first-order interim rule with respect to a matroid constraint
admits a dice-based implementation.
This illustrates that the structure revealed by Myerson's
characterization is more general, and applies to other settings 
in which only first-order interim information is relevant.
For example, single-item auctions with (public or private) budgets are
such a setting (see, e.g., \cite{pai:vohra}).
Our result also provides a generalization of order sampling from the
$k$-winner setting to the general matroid, multi-type setting.

Beyond the existential proof of dice-based implementations,
we show (in Section~\ref{sec:interim-to-dice})
that for single-winner environments,
an algorithm can construct the dice-based rule \emph{efficiently}.
When the types are identically distributed, we also constructively
show (in \camera{Section~\ref{sec:iid}}{the full version})
that every first-order interim rule
which is symmetric across candidates admits a symmetric dice implementation;
i.e., different candidates have the same die for the same type.
This is consistent with Myerson's symmetric characterization of
optimal single-item auctions with i.i.d.~bidders, and generalizes it
to any other first-order single-winner selection setting in which
candidates are identical.
Single-item auctions with identically distributed budgeted bidders are
such a setting, and a symmetric dice-based implementation of the
optimal allocation rule was already known from \citep{pai:vohra}.

For single-winner cases, we also show the converse direction:
how to efficiently compute the first-order
interim rule of a given dice-based winner selection rule.
In effect, these results show that collections of dice are a
computationally equivalent description of single-winner first-order
interim information.
This implies a kind of equivalence between the two dominant approaches
for mechanism design: the Myersonian approach based on virtual values
(i.e., dice), and the Borderian approach based on optimization over interim rules. 

In an attempt to leverage the same kinds of insights for
  Bayesian persuasion, we examine
(in Sections~\ref{sec:dice-to-interim} and \ref{sec:persuasion})
the dice implementability of second-order interim rules.
When the candidate type distributions are non-identical,
we show an impossibility result.
We construct an instance of Bayesian persuasion with independently
distributed non-identical actions,
and show that no optimal persuasion scheme for this instance can be
implemented by dice.
Since second-order interim rules are sufficient for evaluating the
objective and constraints of Bayesian persuasion,
this implies that there exist second-order interim rules which are not
dice-implementable.
This rules out the Myersonian approach for characterizing and
computing optimal schemes for Bayesian persuasion 
with independent non-identical actions, complementing the negative
result of \cite{dughmi2016algorithmic} which rules out the Borderian
approach for the same problem. 

Our impossibility result disappears when the actions are i.i.d.,
since second-order interim rules collapse to first-order interim rules in
symmetric settings.
In particular, as we show in Section~\ref{sec:persuasion},
our results for first-order interim rules,
combined with those of \cite{dughmi2016algorithmic},
imply that Bayesian persuasion with i.i.d.~actions admits an optimal
dice-based scheme, which can be computed efficiently.

\subsection{Additional Discussion of Related Work}
\label{sec:related}
\citet{myerson} was the first to characterize revenue-optimal single-item auctions; this characterization extends to single-parameter mechanism design settings more generally (see, e.g., \cite{hartlinebook}). 
The (first-order) interim rule of an auction, also known as its reduced form, was first studied by \citet{maskin_riley_84} and \citet{matthews_84}. The inequalities characterizing the space of feasible interim rules were described by \citet{border,border2007reduced}. Border's analytically tractable characterization of feasible interim rules has served as a fruitful framework for mechanism design, since an optimal auction can be viewed as the solution of a linear program over the space of interim rules. Moreover, this characterization has enabled the design of efficient algorithms for recognizing interim rules and optimizing over them, by \citet{cai2012algorithmic} and \citet{alaei12}. This line of work has served as a foundation for much of the recent literature on Bayesian algorithmic mechanism design in multi-parameter settings.

It is important to contrast our dice-based rule with the characterization of \citet{cai2012algorithmic}. In particular, the results of \citet{cai2012algorithmic} imply that every first-order interim rule can be efficiently implemented as a distribution over virtual value maximizers. In our language, this implies the existence of an efficiently computable dice-based implementation \emph{in which the dice may be arbitrarily correlated}. Our result, in contrast, efficiently computes a family of \emph{independent dice} implementing any given first-order interim rule in single-winner settings, and shows the existence of a dice-based rule in matroid settings. This is consistent with Myerson's characterization, in which virtual values are drawn independently.
  
\citet{alaei12} also studied winner-selection environments, under the
different name ``service based environments.'' For single-winner
settings, they proposed a mechanism called stochastic sequential
allocation (SSA). The mechanism also implements any feasible
first-order interim rule, by creating a token of winning and
transferring the token sequentially from one candidate to another,
with probabilities defined by an efficiently computed transition
table.
Dice can be considered as the special case of SSA in which the
  transition probabilities are independent of the current owner of the token.
  
As another motivation for our focus on dice-based rules,
order sampling studies how to sample $k$ winners from $n$ candidates
with given inclusion probabilities (i.e., implement an interim rule),
by assigning a random score variable (die) to each candidate.
\citet{rosen1997asymptotic} showed that parameterized Pareto
distributions can be used to implement a given interim rule asymptotically.
\citet{aires2002order} proved the existence of an order sampling
scheme that exactly implements any feasible interim rule.
Our existential proof is a generalization of the proof of \cite{aires2002order} to settings with multiple types and matroid constraints. 

The Bayesian persuasion model is due to \citet{Kamenica11}, and is the most influential model in the general space of information structure design (see the survey by \citet{shaddin_ISD_survey} for references). Bayesian persuasion was examined algorithmically by \citet{dughmi2016algorithmic}, who observed its connection to auction theory and interim rules, and examined the computational complexity of optimal schemes through the lens of optimization over interim rules. 

Of particular relevance to our work is the negative result of \citet{dughmi2016algorithmic} for Bayesian persuasion with independent non-identical actions: it is \#P-hard to compute the interim rule (first- or second-order) of the optimal scheme, or more simply even the sender's optimal utility. Most notable about this result is what it \emph{does not} rule out: an algorithm implementing the optimal persuasion scheme ``on the fly,'' in the sense that it efficiently samples the optimal scheme's (randomized) recommendation when given as input the profile of action types. Stated differently, the negative result of \citet{dughmi2016algorithmic} merely rules out the Borderian approach for this problem, leaving other approaches --- such as the Myersonian one --- viable as a means of obtaining an efficient ``on the fly'' implementation.  This would not be unprecedented: \citet{gopalan} exhibit a simple single-parameter auction setting for which the optimal interim rule is \#P hard to compute, yet Myerson's virtual values can be sampled efficiently and used to efficiently implement the optimal auction. Our negative result in Section~\ref{sec:persuasion} rules out such good fortune for Bayesian persuasion with independent non-identical actions: there does not exist a (Myersonian) dice-based implementation of the optimal persuasion scheme in general.


\section{Preliminaries} \label{sec:preliminaries}

\subsection{Winner Selection}
Consider choosing a set of winners from among \NUMC \emph{candidates}.
Each candidate $i$ has a type $\type{i} \in T_i$,
drawn independently from a distribution \PMF[i].
A winner-selection rule \ALGO maps each type profile
$\TP = (\type{1}, \ldots, \type{\NUMC})$, possibly randomly,
to one of a prescribed family of feasible sets $\I\subseteq 2^{[n]}$.
When $i\in \ALGO(\TP)$, we refer to $i$ as a
\emph{winning candidate}, and to $t_i$ as his \emph{winning type}.
Writing $\PMF = \PMF[1] \times \cdots \times \PMF[n]$ for the
  (independent) joint type distribution,
we also refer to $(\PMF, \I)$ as the \emph{winner-selection environment}. 
When $\I$ is the family of singletons, as in the setting of the single
item auction,
we call $(\PMF,\I)$ a \emph{single-winner environment}.

This general setup captures the allocation rules of general
auctions with independent unit-demand buyers,
albeit without specifying payment rules or imposing incentive constraints.
Moreover, it captures Bayesian persuasion with independent action
payoffs, albeit without enforcing persuasiveness (also called \emph{obedience})
constraints.


\subsection{Matroids}

In this paper, we focus on settings in which the feasible sets $\I$ are
the independent sets of a matroid.
We use the standard definition of a matroid $\M$ as a pair $(E, \I)$,
where $E$ is the \emph{ground set} and $\I\subseteq 2^E$ is a family
of so-called \emph{independent sets}, satisfying the three matroid axioms.
We also use the standard definitions of a \emph{circuit} and
\emph{rank function} $r_\M : 2^E \to \NN$.
The \emph{restriction} $\M|S$ of $\M=(E,\I)$ to some $S \sse E$ is the
matroid $(E, \I \intersect 2^S)$.\footnote{Note that we deviate
  slightly from the standard definition in that we do not restrict the
  ground set.}
For details on matroids, we refer the reader to \citet{oxley2006matroid}. 

A matroid $\M = (E, \I)$ is \emph{separable} if it is a \emph{direct sum}
of two matroids $\M_1=(E_1, \I_1)$ and $\M_2=(E_2, \I_2)$.
Namely, $E = E_1\uplus E_2$, $\I = \Set{A\cup B}{A\in \I_1, B\in \I_2}$. 
Note that if $\M$ is non-separable, then $r_\M(E) < |E|$;
otherwise $\M$ is the direct sum of singleton matroids.
We use the following theorem.

\begin{theorem}[\citet{whitney1935abstract}]
  (1) When $\M=(E, \I)$ is a non-separable matroid, for every $a,b\in E$, there is a circuit containing both $a$ and $b$.
  (2) Any separable matroid $\M$ is a direct sum of two or more
  non-separable matroids called the \emph{components} of $\M$. 
\end{theorem}

In most of the remainder of the paper, we focus on winner selection
environments $(\PMF,\I)$ where $\I$ is the family of independent sets of
a matroid $\M$.
We therefore also use $(\PMF,\M)$ to denote the environment.

\subsection{Interim Rules and Border's Theorem} \label{prelim:interim}
A \emph{(first-order) interim rule} \INTE specifies the winning probability
$\interim[i]{t} \in [0,1]$ for all $i \in [\NUMC], t \in T_i$ in an
environment $(\PMF, \I)$.
More precisely, we say that a winner-selection rule \ALGO
\emph{implements} the interim rule \INTE for a prior \PMF if it
satisfies the following:
if the type profile $\TP = (\type{1}, \ldots, \type{n})$
is drawn from the prior distribution
$\PMF = \PMF[1] \cross \cdots \cross \PMF[\NUMC]$,
then $\ProbC{i\in \ALGO[\TP] }{\type{i} = t} = \interim[i]{t}$.
An interim rule is \emph{feasible} (or \emph{implementable}) within an
environment $(\PMF, \I)$
if there is a winner-selection rule implementing it that always
outputs an independent set of $\I$.

\subsubsection{Border's theorem and implications for  single-winner environments}
The following theorem characterizes the space
of feasible interim rules for single-winner settings.

\begin{theorem}[\citet{border,border2007reduced}]
  \label{thm:border}
An interim rule \INTE is feasible for a single-winner setting if and only if for all possible type
subsets
$S_1\subseteq T_1, S_2\subseteq T_2, \ldots,
 S_n\subseteq T_{\NUMC}$,  
\begin{equation}
  \sum_{i=1}^\NUMC \sum_{t \in S_i} \pmf[i]{t} \interim[i]{t}
  \le 1 - \prod_{i=1}^\NUMC \left( 1 - \sum_{t\in S_i} \pmf[i]{t} \right).
\label{eqn:border-constraint}
\end{equation}
\end{theorem}

The following result leverages Theorem~\ref{thm:border} to show that
efficient algorithms exist for checking the feasibility of an interim
rule, and for implementing a feasible interim rule. 

\begin{theorem}[\citep{cai2012algorithmic,alaei12}] \label{thm:borderpolytime}
Given explicitly represented priors $\PMF[1], \ldots, \PMF[\NUMC]$ and
an interim rule \INTE in a single-winner setting, the feasibility of \INTE can be checked
in time polynomial in the number of candidates and types.
Moreover, given a feasible interim rule \INTE,
an algorithm can find a winner-selection rule implementing \INTE in time
polynomial in the number of candidates and types. 
\end{theorem}


In our efficient construction for single-winner settings, we utilize a structural result  
which shows that checking only a subset of Border's
constraints suffices \cite{border,mierendorff2011asymmetric,cai2012algorithmic}.
This subset of constraints can be identified efficiently. 

\begin{theorem}[Theorem 4 of \cite{cai2012algorithmic}] \label{thm:prefix}
An interim rule \INTE is feasible for a single-winner setting if and only if for all possible
$\alpha \in [0,1]$, the sets
$S_i(\alpha) = \Set{t \in T_i}{\interim[i]{t} > \alpha}$
satisfy the following Border's constraint: 
\[
  \sum_{i=1}^{\NUMC} \sum_{t\in S_i(\alpha)} \pmf[i]{t} \interim[i]{t}
  \leq 1 - \prod_{i=1}^{\NUMC} \left(1 - \sum_{t\in S_i(\alpha) } \pmf[i]{t}\right).
\]
\end{theorem}

When the candidates' type distributions are i.i.d., i.e.,
$T_i$ and \PMF[i] are the same for all candidates $i$,
it is typically sufficient to restrict attention to \emph{symmetric}
interim rules.
For such rules, \interim[i]{t} is equal for all candidates $i$.
In the i.i.d.~setting, we therefore notationally omit the dependence on
the candidate and let $T$ refer to the common type set of all candidates,
\PMF to the candidates' (common) type distribution,
and \interim{t} to the probability that a particular candidate wins
conditioned on having type $t$.
In the i.i.d.~setting, only the symmetric constraints from
Theorem~\ref{thm:border} suffice to characterize feasibility
\citep{border}; namely,
\begin{equation}
\NUMC \cdot \sum_{t\in S} \pmf{t} \interim{t}
  \; \le \; 1 - \left( 1-\sum_{t \in S} \pmf{t} \right)^{\NUMC} ,
\label{eqn:symmetric-border}
\end{equation}
for all $S \sse T$.
Theorem~\ref{thm:prefix} then implies that it suffices to check
Inequality~\eqref{eqn:symmetric-border} for sets of the form
$S(\alpha) = \set{ t \in T : \interim{t} > \alpha}$ with $\alpha \in [0,1]$. 

\subsubsection{Border's Theorem for matroid environments}
For general settings with matroid constraints, \citet{alaei12}
established the following generalized ``Border's Theorem.''

\begin{theorem}[Theorem 7 of \citep{alaei12}]
  Let $\candidate$ map each type $t$ to the (unique) candidate $i$ with
  $t \in T_i$.
  An interim rule \INTE is feasible within an environment $(\PMF, \M)$
  if and only if for all possible type subsets
  $S_1\subseteq T_1, S_2\subseteq T_2, \ldots, S_n\subseteq T_{\NUMC}$,
\camera{%
\[
    \sum_{i=1}^n \sum_{t\in S_i}\pmf[i]{t} \interim[i]{t}
    \le \Expect[\vc{t} \sim \PMF]{r_{\M}(\candidate(\vc{t}\cap S))}, 
\]}{%
  $
    \sum_{i=1}^n \sum_{t\in S_i}\pmf[i]{t} \interim[i]{t}
    \le \Expect[\vc{t} \sim \PMF]{r_{\M}(\candidate(\vc{t}\cap S))}, 
  $}
  where $S=\bigcup_{i=1}^n S_i$.
\end{theorem}
In later sections, we omit the function $\candidate$,
and for any type set $S$ just write $r_\M(S)$ instead of $r_\M(\candidate(S))$.

\subsection{Winner-Selecting Dice}

We study winner-selection rules based on \emph{dice}, as a generalization of order sampling to multiple types and
general constraints.
A \emph{dice-based rule} fixes, for each type $t\in T_i$,
a distribution \die[i]{t} over real numbers, which we call a \emph{die}.
Given as input the type profile $\TP = (\type{1}, \ldots, \type{\NUMC})$,
the rule independently draws a score
$v_i \sim \die[i]{\type{i}}$ for each candidate $i$ by ``rolling
his die;'' it then selects the feasible set of candidates maximizing
the sum of scores as the winner set, breaking ties with a predefined rule. 
In this paper, we will mainly discuss matroid feasibility constraints,
for which a feasible set maximizing the sum of scores can be found by
a simple greedy algorithm:
candidates are added to the winner set in decreasing order of their
scores, breaking ties uniformly at random, as long as the new winner
set is still an independent set of the matroid and their scores are
positive.
When candidates have the same type sets,
we call a dice-based rule \emph{symmetric} if \die[i]{t} is the same
for all $i$. 


Myerson's optimal auction is a dice-based winner-selection rule.
In Myerson's nomenclature, \score{i} is candidate $i$'s \emph{virtual value},
and \die[i]{t} is a single-sided die with the virtual value.


Let $T$ be the set of all types of all candidates and
$\DIESET = (\DIE[t])_{t \in T}$ be a vector of dice, one per type.
Given an interim rule \INTE, and a winner-selection environment
$(\PMF, \I)$, we say that \DIESET \emph{implements} \INTE,
or \DIESET describes \emph{winner-selecting dice} for \INTE in
$(\PMF, \I)$, if the dice-based rule given by \DIESET implements
\INTE within the environment $(\PMF, \I)$.

\subsection{Second-order Interim Rules}
\label{prelim:secondorder_interim}

A (first-order) interim rule, as defined in Section~\ref{prelim:interim},
specifies, for each candidate $i$, the conditional type distribution
of $i$ in the event that $i$ is chosen as the winner.
We define a \emph{second-order interim rule}%
\footnote{Our notion of second-order interim rules is
different from the notion defined in \cite{cai2012optimal}.
Because \citet{cai2012optimal} consider correlation in types,
their notion of second-order interim rules is aimed at capturing the
allocation dependencies arising through such type
correlation, rather than solely through the mechanism's choice.}
which maintains strictly more information,
as needed for describing the incentive constraints of Bayesian persuasion.
Such a rule specifies, for each pair of candidates $i$ and $i'$
(where $i'$ may or may not be equal to $i$),
the conditional type distribution of $i'$ in the event that $i$
is chosen as the winner.
Formally, a second-order interim rule \SINTE specifies
$\sinte{i}{i'}{t} \in [0,1]$ for each pair of candidates 
$i, i' \in [\NUMC]$, and type $t \in T_{i'}$.
We say that a winner-selection rule \ALGO 
\emph{implements} \SINTE for a prior \PMF if it satisfies the following:
if the type profile $\TP = (\type{1}, \ldots, \type{\NUMC})$
is drawn from the prior distribution
$\PMF=\PMF[1] \cross \cdots \cross \PMF[\NUMC]$,
then $\ProbC{i\in\ALGO(\TP) }{\type{i'} = t} = \sinte{i}{i'}{t}$.
A second-order interim rule is \emph{feasible} if there is a winner
selection rule implementing it.

\section{Existence of Dice Implementation for Matroids}\label{sec:exist}
\newcommand{\inductionpoint}{\theta_k, \suf_{k+1}}
\newcommand{\BT}{\ensuremath{K}\xspace}

\camera{In this section, we prove our first theorem:}{In this section,
  we outline the proof of our first theorem:}

\begin{theorem}\label{thm:finite-face}
Let $(\PMF, \M)$ be a matroid winner selection-environment
with a total of $m$ types,
and let \INTE be an interim rule that is feasible within $(\PMF, \M)$.
There exist winner-selecting dice $\D$, each of which has at
most $m+1$ faces, which implement \INTE.
\end{theorem}

The proof consists of two parts.
First, we generalize the result of \cite{aires2002order},
which showed the existence of continuous winner-selecting dice for feasible
interim rules for a $k$-uniform matroid with fixed types, to general
matroids and multiple types. 
Second, we convert the continuous dice to dice with at most $m+1$
faces each, while keeping the interim probabilities unchanged.

\subsection{Continuous Winner-Selecting Dice}
\begin{theorem}\label{thm:exist}
  Let $\M$ be a matroid, and $\INTE$ a feasible interim rule within the
  winner-selection environment $(\PMF, \M)$.
  There exist winner-selecting dice $\D$ over $\mathbb{R}$
  that implement $\INTE$ in $(\PMF, \M)$. 
\end{theorem}

We assume without loss of generality that the candidates' type sets
are disjoint, 
and use $T =\biguplus_{i=1}^{\NUMC} T_i$ to denote the set of all types.
We use \pmf{t} and \interim{t} as shorthand for \pmf[i]{t} and \interim[i]{t},
where $i = \candidate(t)$ is the candidate for whom $t \in T_i$.
Moreover, given a set of types $S \sse T$,
we write $S_i = S \intersect T_i$. 
Recall the Border constraints 
\begin{equation}
  \label{eq:border-matroid}
  \sum_{i=1}^n \sum_{t\in S_i}\pmf[i]{t} \interim[i]{t} \le R(S),   
\end{equation}
where $R(S)=\Expect[\vc{t}\sim \PMF]{r_{\M}(\vc{t}\cap S)}$ is the
expected rank of types in $S$ which show up,
a submodular function over the type set $T$. 
The Border constraints can therefore be interpreted as follows: 
An interim rule $\INTE: T \to [0,1]$ is feasible for \PMF and $\M$
if and only if $\tilde{\INTE}$ is in the polymatroid given by
$R(S)$, where $\tilde{x}(t) := \pmf{t} \interim{t}$.
Equivalently, \INTE is feasible if and only if the submodular
\emph{slack function}
$\slack(S) = R(S) - \sum_{t \in S} \pmf{t} \interim{t}$
is non-negative everywhere.

When \INTE is feasible, we call a set $S \sse T$ \emph{tight}
for $(f, \INTE, \M)$ if the Border constraint \eqref{eq:border-matroid} corresponding to $S$ is tight at $\INTE$,
i.e., $\sum_{i=1}^n \sum_{t\in S_i}\pmf{t} \interim{t} = R(S)$.
By definition, $S=\emptyset$ is always tight. 
The family of tight sets, being the family of minimizers of the 
submodular slack function, forms a lattice:
the intersection and the union of two tight sets is a tight set. 

\begin{remark}\label{remark:tight}
  The tightness of a set $S$
  means that the expected number of winners from $S$ equals the expected
  rank of  types in $S$ which show up. In other words, $S$ is tight if
  and only if a maximum independent subset of $\vc{t}\cap S$ 
  is always selected as winners.
\end{remark}

By \camera{Remark~\ref{remark:tight}}{the preceding remark},
the types in minimal non-empty tight sets need to be treated preferentially,
i.e., assigned higher faces on their dice, compared to types outside them.
Because they play such an important role, we define them as \emph{barrier sets}.
Formally, we define the set of \emph{active} types
$T^+= \set{ t \in T: \pmf{t} \interim{t} > 0}$
to be the types who win with positive probability.
Barrier sets are subsets of $T^+$.
If there is at least one non-empty tight set of active types,
we define the barrier sets as the (inclusion-wise) minimal such sets.
Otherwise, we designate the entire set $T^+$ of active types as the
unique barrier set.

To prove Theorem~\ref{thm:exist}, we first assume that the matroid $\M$ is non-separable. Separable matroids will be handled in the proof of Theorem~\ref{thm:exist} by combining the dice of their non-separable components. Because barrier sets get precedence, we first show how to construct dice for barrier sets with Lemma~\ref{lem:exist-barrier}. Once we have dice for barrier sets, we can repeatedly ``peel off'' the tight sets and combine their dice, which is captured in Lemma~\ref{lem:exist-lattice}. We start with the existence of dice for barrier sets:

\begin{lemma}\label{lem:exist-barrier}
  Let $\M$ be a non-separable matroid, and 
  $\INTE > \vc{0}$ a feasible interim rule within the
  winner-selection environment $(\PMF, \mathcal{M})$.
  Let $S$ be a barrier set for $(\PMF, \INTE, \M)$, and define
  $\INTE_S$ to be  $(\interim{t})_{t\in S}$.
  There exists a vector of 
  distributions $\D=(\DIE[t])_{t\in S}$ over $\mathbb{R}$,
  such that $\D$ implements $\INTE_S$ in $(\PMF, \M|S)$. 
\end{lemma}

\camera{%
Without loss of generality, assume that the barrier set $S$ contains
\BT types, and number them by $1, 2, \dots, \BT$.

\begin{proof}
  When $|\candidate(S)|=1$ and $S$ is tight, all types in $S$ belong
  to the same candidate, so there is no competition
  between candidates. Because of the tightness of $S$, we can simply
  assign a single-sided die with face value $1$ to each type in $S$. Thus in the
  proof, we assume that either there is more than one candidate with a
  type in $S$ or that $S$ is not tight.%
\footnote{Recall that $S$ could be the (possibly non-tight) barrier
  set of all active types, if no non-empty set is tight.}
In both cases, this implies that $\interim{t} < 1$ for all $t\in S$.
For if $\interim{t} = 1$, the singleton $\SET{t}$ would
be a proper tight subset of $S$,
which contradicts the assumption that $S$ is a barrier set.

We create \DIE[t] from any continuous distribution%
\footnote{For example, one can use an exponential distribution.}
$D$ with support $[0, \infty)$.
For each type $t$, we assign a parameter $\theta_t > 0$.
Also, we choose a global parameter $\tau$.
To sample from \DIE[t],
we draw a \emph{primitive roll} $v_t\sim D$ and output
$\theta_tv_t - \tau$. 

We define $g_t(\tau, \theta_1, \dots, \theta_\BT)$ to be the
interim probability of type $t$ winning in the environment $(\PMF,
\M|S)$ when the winner-selecting dice \DIE[t] with parameters
$\theta_t$ and $\tau$ are used for $t\in S$.
Each $g_t$ is continuous in all of its parameters.
Using $\vc{\theta}_S = (\theta_t)_{t\in S}$ for short,
consider the following system of equations with variables $\vc{\theta}_S$:
  \begin{align}
  \sum_{t\in S} \interim{t} g_t(\tau, \vc{\theta}_S)
    &= \sum_{t\in S} \pmf{t} \interim{t}\label{eq:interim-system-sum}\\
  g_t(\tau, \vc{\theta}_S)
    &= \interim{t}, \quad \mbox{ for all } t\in [\BT]. \label{eq:interim-system}
  \end{align}
  The objective of the proof is to show that the system of equations
  \eqref{eq:interim-system} admits a solution.
  This solution will directly induce a dice system that implements
  $\INTE$ for types in $S$.  
  Notice that the system is redundant:
  Equation~\eqref{eq:interim-system-sum} is implied by the Equations~\eqref{eq:interim-system}.
  Throughout the proof, we use $\pre_i$ to denote the prefix
  $(\tau, \theta_{1},\dots, \theta_{i})$ and $\suf_i$ to denote the suffix
  of parameters $(\theta_{i}, \theta_{i+1}, \dots, \theta_{n})$. In particular, $\pre_0 = \tau$.
  We prove the following inductively for $k=1,\dots, \BT$:
  \begin{quote}
    For any
    positive suffix $\suf_k > \vc{0}$, there is a prefix $\pre_{k-1}$ so
    that $(\pre_{k-1}(\suf_k), \suf_k)$ satisfies the first $k$ equations
    in the system
    \eqref{eq:interim-system-sum}-\eqref{eq:interim-system}. 
  \end{quote}
  
  Applying the claim with $k=\BT$, for every positive $\theta_\BT$, there is a
  solution that satisfies the first \BT equations.
  Because of the redundancy of the system, satisfying the first \BT equations
  guarantees that the last equation is also satisfied. 

  For the base case $k=1$,
  consider an arbitrary positive $\suf_1 = \vc{\theta}_S$.
  We need to prove the existence of a $\tau$ such that
  $\sum_{t\in S} \pmf{t} g_t(\tau, \vc{\theta}_S)
  = \sum_{t\in S} \pmf{t} \interim{t}$.
  When $\tau = 0$, all types get non-negative die rolls.
  Therefore, a maximum independent subset of $\vc{t}\cap S$ is always
  selected.
  Thus, $S$ is tight at $(g_t(0, \vc{\theta}_S))_{t\in S}$, i.e.,
  $\sum_{t\in S} \pmf{t} g_t(0, \vc{\theta}_S) = R(S)
  \ge \sum_{t\in S} \pmf{t} \interim{t}$.
  On the other hand, as $\tau$ increases,
  the probability that all die rolls are negative goes to 1,
  meaning that in the limit, no agent wins.
  Thus, $\lim_{\tau\to \infty} g_t(\tau, \vc{\theta}_S) = 0$.
  Furthermore, each $g_t(\tau, \vc{\theta}_S)$
  strictly and continuously decreases with $\tau$.
  Because $0 < \sum_{t\in S} \pmf{t} \interim{t} \leq R(S)$,
  by the Intermediate Value Theorem,
  for every  $\vc{\theta}_S$, there is a unique $\tau$ such that
  $\sum_{t\in S} \pmf{t} g_t(\tau, \vc{\theta}_S)
  = \sum_{t\in S} \pmf{t} \interim{t}$.
  We denote this unique $\tau$ by $h_0(\vc{\theta}_S)$;
  notice that $h_0(\vc{\theta}_S)$ is a continuous function of $\vc{\theta}_S$.
  When $S$ is a tight barrier set, i.e., $R(S) = \sum_t f(t) x(t)$,
  the equation is satisfied for $\tau = 0$, so $h_0(\vc{\theta}_S) = 0$.
  This establishes the base case $k=1$ of the induction hypothesis.  

  For the inductive step, fix an arbitrary $k \ge 1$,
  and let $\suf_{k+1} > \vc{0}$ be arbitrary.
  For any fixed $\theta_k > 0$, by induction hypothesis,
  there is a unique $\pre_{k-1} (\theta_k, \suf_{k+1})$ such that
  $(\pre_{k-1} (\theta_k, \suf_{k+1}), \theta_k, \suf_{k+1})$
  satisfies the first $k$ equations
  of \eqref{eq:interim-system-sum}--\eqref{eq:interim-system}.
  Lemma~\ref{lem:range}
  below shows that there is a unique $\theta_{k} = h_k(\suf_{k+1})$
  such that
  $(\pre_k (h_k(\suf_{k+1}), \suf_{k+1}), h_k(\suf_{k+1}), \suf_{k+1})$
  satisfies the first $k+1$ equations of
\eqref{eq:interim-system-sum}-\eqref{eq:interim-system}.
The inductive claim now follows by defining
$\pre_{k+1}(\suf_{k+1}) = (\pre_k(h_k(\suf_{k+1}),\suf_{k+1}), h_k(\suf_{k+1})).$
\end{proof}

  
\begin{lemma}\label{lem:range}
  Fix any feasible \INTE.
  For every $k$ and suffix $\suf_{k+1}$,
  there is a unique $\theta_k\in (0, \infty)$ such that
  $g_k(\pre_{k-1}(\theta_k, \suf_{k+1}), \theta_k, \suf_{k+1}) = \interim{k}$.
\end{lemma}

The proof of Lemma~\ref{lem:range} is based on large part on the
following monotonicity properties:

\begin{lemma}\label{lem:monotone-reduced}
  Assume that $S$ contains types of more than one candidate or is not
  tight.
  For any $t \ge k$, consider
  $g_t(\pre_{k-1}(\theta_k, \suf_{k+1}), \theta_k, \suf_{k+1})$
  as a function of $\theta_i$, for $i \ge k$.
  The function $g_t$ satisfies the following properties:
  \begin{enumerate}
  \item It is weakly decreasing in $\theta_{i}$,
    for all $i \ne t, i \ge k$.
  \item It is strictly increasing in $\theta_t$.
  \end{enumerate}
\end{lemma}

\begin{lemma} \label{lem:inter-candidate}
If $\candidate(t) \ne \candidate(i)$, then $g_t(\vc{\theta}_S)$ is
strictly decreasing in $\theta_{i}$.
If in addition $i \geq k$ and $t \geq k$, then 
$g_t(\pre_{k-1}(\suf_{k}), \suf_{k})$
is also strictly decreasing in $\theta_{i}$.
\end{lemma}

For notational convenience, we define the shorthand
$g_t(\inductionpoint) :=
g_t(\pre_{k-1}(\theta_k, \suf_{k+1}), \theta_k, \suf_{k+1})$.
Furthermore, for fixed $\suf_{k+1}$ (which will be clear from the context),
we define
$h_t(\theta_k)$ to be the \Kth{t}
component of $\pre_{k-1}(\theta_k, \suf_{k+1})_t$, for any $t < k$.

\begin{extraproof}{Lemma~\ref{lem:range}}
As we did in the proof of Lemma~\ref{lem:exist-barrier} for the base
case of the induction claim,
we will examine the limits of $g_k(\inductionpoint)$ as
$\theta_k \to 0$ and $\theta_k \to \infty$.
By doing so, we will establish that $\interim{k}$
lies between the two limits.
Then, using the continuity and strict monotonicity
(by Lemma~\ref{lem:monotone-reduced}) of $g_k(\inductionpoint)$,
the Intermediate Value Theorem implies that there is a unique $\theta_k$
satisfying $g_k(\inductionpoint) = \interim{k}$.

\begin{itemize}
\item To compute $\lim_{\theta_k\to \infty} g_k(\inductionpoint)$,
  consider the set of types
  $A = \Set{t < k}{\lim_{\theta_k\to \infty} h_t(\theta_k)=\infty}$.
  As $\theta_k \to \infty$,
  with probability approaching $1$,
  $A\cup\SET{k}$ will dominate all other types. 
  We distinguish two cases, based on whether
  $\lim_{\theta_k \to \infty} h_0(\vc{\theta}_S)$ is finite or not. 
\begin{itemize} 
\item If $\lim_{\theta_k\to \infty}h_0(\theta_{k}) \neq \infty$,
  the parameter $\tau$ will be finite in the limit. Therefore,
  a maximum independent subset of $A \cup \SET{k}$ will be selected as
  winners,  implying that $A \cup \SET{k}$ is tight in the limit.
  Because $\pre_{k-1}(\theta_k, \suf_{k+1})$ ensures that the first
  $k$ equations of
  \eqref{eq:interim-system-sum}--\eqref{eq:interim-system} are satisfied,
  each type $t \in A$ wins with probability $\pmf{t} \interim{t}$.
  Because $A \cup \SET{k}$ is tight, 
  $\lim_{\theta_k \to \infty} \pmf{k} g_k(\inductionpoint)
  = R (A\cup \SET{k}) - \sum_{t\in A} \pmf{t} \interim{t}$.
  From the Border constraint for the set $A\cup \SET{k}$,
$\pmf{k} \interim{k} + \sum_{t\in A} \pmf{t} \interim{t} < R(A\cup\SET{k})$;
otherwise $A \cup \SET{k}\subsetneq S$ would be a tight set,
contradicting that $S$ is a barrier set.
Rearranging, we have shown that
$\lim_{\theta_k \to \infty} \pmf{k} g_k(\inductionpoint) > \pmf{k} \interim{k}$.
\item If $\lim_{\theta_k \to \infty} h_0(\theta_{k}) = \infty$,
  then all types $t > k$ will get negative die rolls with probability
  approaching 1, so $\lim_{\theta_k\to\infty} g_t(\inductionpoint) = 0$
  for all types $t > k$.
  Thus,
  $\lim_{\theta_k\to\infty} \pmf{k} g_k(\inductionpoint)
  + \sum_{t<k} \pmf{t} \interim{t} = \sum_{t\in S} \pmf{t} \interim{t}$,
  i.e.,
  $\lim_{\theta_k\to\infty}\pmf{k} g_k(\inductionpoint)
  = \sum_{t\ge k} \pmf{t} \interim{t} > \pmf{k} \interim{k}$,
  because all types in a barrier set are active by definition.
\end{itemize}



\item To compute $\lim_{\theta_k \to 0}g_k(\inductionpoint)$, define
  the type set
  $A = \Set{1 \le t < k}{\lim_{\theta_k\to 0} h_t(\theta_k) > 0}
  \cup \SET{k+1, \dots, \BT}$.
  As $\theta_k\to 0$, with probability approaching $1$,
  $A$ will dominate all other types because for
  $t \notin A$, $h_t(\theta_k)$ goes to $0$.
  Again, we distinguish two cases,
  based on whether $\lim_{\theta_k \to \infty} h_0(\theta_k) = 0$ or not. 
\begin{itemize}
\item If $\lim_{\theta_k \to \infty} h_0(\theta_k) = 0$,
  then in the limit, a maximum independent subset of $A$ must always be
  chosen as winners, so $A$ approaches tightness. 
  Again, $\pre_{k-1}(\theta_k, \suf_{k+1})$ ensures that
  $g_{t}(\inductionpoint) = \interim{t}$ for $t < k$, so we have
  $\lim_{\theta_k \to 0}\sum_{t > k} \pmf{t} g_t(\inductionpoint)
  = R(A) - \sum_{t\in A, t < k} \pmf{t} \interim{t}$.
  Combined with
  $\sum_{t\ge k} \pmf{t} g_t(\inductionpoint)
  = \sum_{t\ge k} \pmf{t} \interim{t}$,
  which is also ensured by $\pre_{k-1}(\theta_k, \suf_k)$,
  we obtain that 
  \[\lim_{\theta_k\to 0} \pmf{k} g_k(\inductionpoint)
    = \sum_{t\ge k} \pmf{t} \interim{t} -
    \left( R(A) - \sum_{t\in A, t<k} \pmf{t} \interim{t} \right). \] 
  Rearranging the Border constraint corresponding to 
  $A\subseteq S \setminus \SET{k}$ gives us that
  $0 > \sum_{t > k} \pmf{t} \interim{t} -
  \left( R(A)-\sum_{t\in A, t < k} \pmf{t} \interim{t} \right)$
  because $A \subsetneq S$ is not tight.
  Adding $\pmf{k} \interim{k}$ on both sides,
  we get $\pmf{k} \interim{k} > \sum_{t \ge k} \pmf{t} \interim{t} -
  \left( R(A) - \sum_{t\in A, t < k}\pmf{t} \interim{t}\right)$.
  Finally, canceling out $\pmf{k}$, we obtain that
  $\interim{k} > \lim_{\theta_k \to 0} g_k(\inductionpoint)$. 
 \item If $\lim_{\theta_k\to 0} h_0(\theta_k) > 0$, as $\theta_k \to 0$,
   type $k$ will get a negative roll with probability
   approaching $1$,
   so $\lim_{\theta_k\to 0} g_k(\inductionpoint) = 0 < \interim{k}$. 
 \end{itemize}
\end{itemize}

In summary, we have shown that
$\lim_{\theta_k \to 0}g_k(\inductionpoint) < \interim{k}
< \lim_{\theta_k\to \infty} g_k(\inductionpoint)$. 
Because  $g_k(\suf_k)$ is continuous and monotone in $\theta_k$,
by the Intermediate Value Theorem,
there exists a unique $\theta_k$ such that $g_k(\suf_k) = \interim{k}$.
\end{extraproof}

In the proofs for Lemmas~\ref{lem:monotone-reduced} and
\ref{lem:inter-candidate},
we will often want to analyze the effect of 
keeping $\theta_{i}$ for all $i \ne t$ unchanged,
while changing $\theta_t$ to $\theta'_t$.
Thereto, we always consider the following coupling of dice rolls in
the two scenarios, which we call \emph{primitive-roll coupling for $t$}:
all dice will obtain the same primitive rolls $\vc{v}$ under the two
scenarios, but scale them differently.
More precisely, we consider the rolls
$(\theta_i v_i - \tau)_{i\in S}$ and $(\theta'_i v_i - \tau')_{i \in S}$,
for any given primitive rolls $\vc{v}$,
where $\theta'_i = \theta_i$ for $i \ne t, i \ge k$.
Recall that all $\theta_i$ for $i < k$ are functions of $\suf_k$. 

\begin{extranoqedproof}{Part 1 of Lemma~\ref{lem:monotone-reduced}}
To show that $g_t(\theta_k, \suf_{k+1})$ is weakly decreasing
in $\theta_i$ for all $i \ge k, i \ne t$,
we again use induction on $k$.
The induction hypothesis is the following:
\begin{enumerate}
\item Each entry of $\pre_{k}(\suf_{k+1})$ is weakly increasing
  in $\theta_i$ for $i > k$.
\item $g_t(\suf_{k+1})$ is weakly decreasing in $\theta_i$
  for $i \ne t, i > k$.
\end{enumerate}

We begin with the base case $k=0$.
The first part of the base case --- that $\pre_0(\vc{\theta}_S)$ is
weakly increasing in $\theta_i$ --- has been shown in the proof of
Lemma~\ref{lem:exist-barrier}. 
To prove the second part of the base case,
we use primitive-roll coupling for $i$ with $\theta_i< \theta'_{i}$. 
For every scenario with primitive rolls $\vc{v}$ in which $t$ is not
a winner with $\vc{\theta}_S$,
we show that $t$ cannot become a winner with $\vc{\theta'}_S$. 
\begin{itemize}
\item Increasing $\theta_i$ to $\theta'_i$ can only (weakly)
  increase the threshold $\tau = \pre_0(\vc{\theta}_S)$,
  by the first part of the base case.
\item Let $A$ be the set of types with rolls higher than the roll
of type $t$ when the rolls are derived from $\vc{v}$ using $\vc{\theta}_S$.
Since $t$ is not winning, we have

\[ 0 \; = \; r_\M(A\cup\SET{t}) - r_\M(A)
     \; \ge \; r_\M(A\cup\SET{t, i}) - r_M(A\cup\SET{i}). \]
The inequality is due to the submodularity of the matroid rank
function, and shows formally that (potentially) adding $i$ to the set
of types with higher rolls than $t$ cannot help $t$ become a winner.
\end{itemize}

\noindent For the induction step, consider some $k \geq 1$ and fix a $\suf_{k+1}$.
By the induction hypothesis,
each entry of $\pre_{k-1}(\theta_k, \suf_{k+1})$ is weakly increasing
in $\theta_i$ for $i \ge k$,
and $g_t(\theta_k, \suf_{k+1})$ is weakly decreasing in $\theta_i$
for $i \ne t, i \ge k$.

\begin{itemize}
\item We first show that each entry of $\pre_k(\suf_{k+1})$
  is weakly increasing in $\theta_i$ for $i > k$.
  The key is component $k$ of $\pre_k(\suf_{k+1})$.
  By the second part of the induction hypothesis, applied with $t=k$,
  $g_k(\pre_{k-1}(\theta_k,\suf_{k+1}), \theta_k, \suf_{k+1})$
  is weakly decreasing in $\theta_i$.
  Since $g_k(\pre_{k-1}(\theta_k,\suf_{k+1}), \theta_k, \suf_{k+1})$
  is defined as the winning probability of $k$ with the given
  parameters, when $\theta_i$ is raised to $\theta'_i > \theta_i$,
  in order to keep
  $g_k(\pre_{k-1}(\theta'_k,\suf'_{k+1}), \theta'_k, \suf'_{k+1})
  = \interim{k} =
  g_k(\pre_{k-1}(\theta_k,\suf_{k+1}), \theta_k, \suf_{k+1})$,
  we require $\theta'_k \geq \theta_k$.
  By applying the induction hypothesis twice, once for $\theta_i$ and
  once for $\theta_k$, 
  all entries of $\pre_{k-1}(\theta_k, \suf_{k+1})$ weakly increase.
  Therefore, all entries of $\pre_{k}(\suf_{k+1})$ are
  weakly increasing in all $\theta_i$ for $i > k$. 
\item To prove the second part of the inductive step,
  recall that
  $g_t(\suf_{k+1}) = g_t(\pre_k(\suf_{k+1}), \suf_{k+1})
    = g_t(h_{k}(\suf_{k+1}), \suf_{k+1})$.
    For $t\le k$, $g_t(\suf_{k+1})=\interim{t}$ is a constant,
    and in particular weakly decreasing in $\theta_i$.
    Consider some $t > k, i > k, t \ne i$.
    First, $h_k(\suf_{k+1})$ is weakly increasing in $\theta_i$
    by the first part of the induction hypothesis.
    By the second part of the induction hypothesis,
    $g_t(\theta_k, \suf_{k+1})$ is weakly decreasing in all of
    its variables except $\theta_t$,
    so substituting $\theta_k = h_k(\suf_{k+1})$ shows that
    $g_t(h_k(\suf_{k+1}),\suf_{k+1})$ is weakly decreasing in $\theta_i$. \QED
\end{itemize}
\end{extranoqedproof}

\begin{extranoqedproof}{Part 2 of Lemma~\ref{lem:monotone-reduced}}
  Recall that $t \geq k$.
  To show that $g_t(\suf_k)$ is strictly increasing in $\theta_t$,
  we will show that at least one type $i \ge k$ has $g_{i}(\suf_k)$
  strictly decreasing in $\theta_t$.
  By Part 1 of the lemma, each $g_{i}(\suf_k)$ for $i \ge k $ is weakly
  decreasing in $\theta_{t}$. 
  By definition, $\pre_{k-1}(\suf_{k})$ ensures that the first $k$
  equations of the system
  \eqref{eq:interim-system-sum}-\eqref{eq:interim-system} are satisfied;
  this implies that
  $\sum_{j\ge k} \pmf{j} g_j(\suf_k) = \sum_{j\ge k} \pmf{j} \interim{j}$.  
  Thus, if at least one of the $g_{i}(\suf_k)$ is strictly decreasing
  in $\theta_t$,
  to keep the summation $\sum_{j\ge k}\pmf{j} g_j(\suf_k)$ unchanged, 
  $g_{t}(\suf_k)$ must increase strictly in $\theta_t$. 
  We consider two possible cases for $S$,
  as permitted by the assumption of the lemma:

  \begin{itemize}
  \item $S$ is not tight.
    In this case, we first show that $h_0(\vc{\theta}_S)$ is a strictly
    increasing function of $\theta_t$.
    Consider the primitive-roll coupling for $\theta'_t > \theta_t$. 
    For any primitive rolls $\vc{v}$, the number of winners will not
    decrease when $\theta_t$ increases to $\theta'_t$.
    Thus, we can bound the summation
    $\sum_{i\in S} \pmf{i} g_i(\tau, \vc{\theta}_S)
    \le \sum_{i\in S} \pmf{i} g_i(\tau, \vc{\theta'}_S)$.
    Because $S$ is not tight, with non-zero probability,
    $\vc{v}$ is such that $v_i\theta_i - \tau < 0$ for all $i$.
    And because $\vc{v}$ is drawn from a continuous distribution,
    with positive probability, $v_t \theta'_t - \tau > 0$ as well.
    In that case, $t$ is the only candidate with a positive die roll.
    This implies that
    $\sum_{i\in S} \pmf{i} g_i(\tau, \vc{\theta}_S)
    < \sum_{i\in S} \pmf{i} g_i(\tau, \vc{\theta'}_S)$.
    Because $h_0(\vc{\theta'}_S)$  is defined as the unique $\tau$
    satisfying the equation
    $\sum_{i\in S} \pmf{i} g_i(\tau, \vc{\theta'}_S)
    = \sum_{i\in S} \pmf{i} \interim{i}$,
    we obtain that $h_0(\vc{\theta'}_S) > h_0(\vc{\theta}_S)$. 

    Now consider any primitive rolls $\vc{v}$ under which a type $i\ne
    t$ wins, and $\vc{v}$ is such that
    $v_{i}\theta_{i} - h_0(\vc{\theta}_S) > 0$ but
    $v_{i}\theta_{i} - h_0(\vc{\theta'}_S) < 0$.
    Such rolls $\vc{v}$ must occur with positive
    probability, because all the dice are fully supported over
    $(0, \infty)$. 
    In that case, $i$ is no longer a winning type under $\vc{\theta'}_S$.
    Thus $g_{i}(\vc{\theta}_S)$  is strictly decreasing in $\theta_t$.
    
  \item There is more than one candidate, i.e.,
     there is a type $i$ with $\candidate(i) \ne \candidate(t)$.
    \begin{itemize}
    \item If there is such a type $i \ge k$
      with $\candidate(i) \ne \candidate(t)$,
      then $g_{i}(\suf_k)$ is strictly decreasing in $\theta_t$
      by the second part of Lemma~\ref{lem:inter-candidate}.
    \item Otherwise, all types $t' \ge k$ have
      $\candidate(t') = \candidate(t)$.
      Define $\suf'_{k} = \suf_{k}$ for all entries except type $t$,
      where it equals $\theta'_t > \theta_t$.
      Consider a type $i < k$ with $\candidate(i) \ne \candidate(t)$.
      By definition of $\pre_{k-1}(\cdot)$, we get that
      $g_{i}(\suf'_k) = \interim{i} = g_{i}(\suf_k)$.
      And by the first part of Lemma~\ref{lem:inter-candidate},
      $g_{i}(\pre_{k-1}(\suf_k), \suf_k) > g_{i}(\pre_{k-1}(\suf_k), \suf'_k)$.
      So $i$ wins with strictly higher probability under
      the parameters $(\pre_{k-1}(\suf'_k), \suf'_k)$ than under
      the parameters $(\pre_{k-1}(\suf_k), \suf'_k)$.
      As shown in Part 1 of Lemma~\ref{lem:monotone-reduced},
      $g_i(\theta_S)$ is weakly decreasing in all of its variable except $\theta_i$. 
      Thus, the only way that the winning probability of agent $i$ can
      increase is for the \Kth{i} component of $\pre_{k-1}(\suf'_k)$
      to be strictly greater than the \Kth{i} component of
      $\pre_{k-1}(\suf_k)$.

      Finally, consider some type $t' \geq k, t' \ne t$ with
      $\candidate(t') = \candidate(t) \ne \candidate(i)$.
      Because all components of $\pre_{k-1}(\cdot_k)$ weakly increase
      going from $\suf_k$ to $\suf'_k$, and the \Kth{i} component
      strictly increases, we get that $g_{t'}(\suf_k) < g_{t'}(\suf'_k)$.
      Thus, we have shown that there is a $t' \ne t$ such that
      $g_{t'}(\suf_k)$ is strictly decreasing in $\theta_t$. \QED
    \end{itemize}
 \end{itemize}
\end{extranoqedproof}

\begin{extraproof}{Lemma~\ref{lem:inter-candidate}}
We prove both statements using primitive-roll coupling.
For the second part of the lemma, notice that when
$\theta_i$ is increased to $\theta'_i$ for $i \geq k$,
all components of $\pre_{k-1}(\suf_k)$ weakly increase.
Thus, all components of $(\pre_{k-1}(\suf_k), \suf_k)$
are weakly larger than those of $(\pre_{k-1}(\suf'_k), \suf'_k)$,
and $\theta'_i > \theta_i$, while $\theta'_t = \theta_t$.
Thus, we can apply primitive-roll coupling in both cases.

Since the matroid is non-separable,
there is a circuit $C$ that contains
$\candidate(t)$ and $\candidate(i)$.
Under the parameter vector $\vc{\theta}_S$,
with non-zero probability (over primitive rolls $\vc{v}$),
the candidates in $C$ get the highest rolls,
$\candidate(t)$ (barely) wins with the second-lowest roll,
and $\candidate(i)$ gets the lowest roll among all candidates in $C$
and does not win.
When $\theta_i$ increases to $\theta'_{i}$,
with non-zero probability, the scaled roll for type $t$ becomes the
lowest so that $t$ ceases to be a winner.
Combining this with Part 1 of Lemma~\ref{lem:monotone-reduced}
(which states that $g_t(\vc{\theta}_S)$ is weakly decreasing in $\theta_i$),
we have that $g_t(\vc{\theta}_S)$ is strictly decreasing in $\theta_{i}$. 
\end{extraproof}
}{The proof is quite technically involved, and due to space
  constraints, it is entirely deferred to the full version.
  The high-level idea is to base the dice upon 
  any full-support continuous distribution for the dice.
  This distribution is scaled by different parameters $\theta_i$ for
  different types $i$, and shifted by a constant $\tau$.
  Matching the prescribed interim winning probabilities \interim{i}
  imposes a system of non-linear equations on the $\theta_i$.
  We establish that the winning probabilities of different types
  satisfy certain key monotonicity properties in the parameters $\theta_i$;
  these monotonicity properties are proved using the fact that the 
  feasible sets form a matroid.
  Then, we apply the Intermediate Value Theorem inductively for all
  types $i$ to non-constructively prove the existence of the desired $\theta_i$.}

To generalize Lemma~\ref{lem:exist-barrier} to arbitrary sets of
types, we will need a construction that allows us to ``scale''
the faces of some dice such that they will always be above/below the faces
of another set of newly introduced dice;
such a construction will allow us to give dice for types in barrier
sets higher faces than other dice.
For the types of full-support distributions over $[0, \infty)$ we have
been using so far, this would be impossible.
There is a simple mapping that guarantees our desired properties:
we map faces from $(0,\infty)$ to the set $(1, 2)$
by mapping all positive $s \mapsto 2 - \frac{1}{1+s}$,
and mapping all negative $s$ to $-1$.
Notice that the new dice implement the same interim rule as the old ones:
in matroid environments, the set maximizing the sum of die rolls is
determined by the greedy algorithm, and hence, only the relative order
between die faces matters.
With the help of this mapping, we prove the following lemma,
similar to Lemma~\ref{lem:exist-barrier}.
\camera{}{The proof is again only given in the full version.}
  
\begin{lemma}\label{lem:exist-lattice}
Let $\INTE > \vc{0}$ be a feasible interim rule within a
winner-selection environment $(\PMF, \M)$,
where $\M$ is a non-separable matroid.
Fix a tight set $S$, and let $\hat{S}$ be a minimal tight
set that includes $S$ as a proper subset, if such a set exists;
otherwise let $\hat{S} = T$.
Given dice $\D =(D_t)_{t\in S}$ that implement $\INTE_S$
in $(\PMF,\M|S)$,
there are dice $\D' = (D'_t)_{t \in \hat{S}}$
which implement
$\INTE_{\hat{S}}$ in $(\PMF,\M|\hat{S})$.
\end{lemma}

\camera{%
\begin{noqedproof}
  First, we define the dice $D'_t$ for types $t \in S$:
  they are obtained by applying the transformation described
  previously to the corresponding dice $D_t$.
  Thus, the range of positive faces is $(1,2)$ for these $D'_t$.

  For types $t \in \hat{S} \setminus S$, 
  we construct new dice $D'_t$ in a similar way to the proof of
  Lemma~\ref{lem:exist-barrier}.
  The construction is essentially identical, but all positive die
  faces $s$ are mapped to $1-\frac{1}{1+s}$.
  Notice that this mapping ensures that all die faces are strictly
  less than 1, and hence, these types $t$ will always lose to types
  $t' \in S$.
  The main other change in the proof is to change the definition of
  $g_t(\tau, \vc{\theta}_{\hat{S} \setminus S})$:
  it is now the interim winning probability of the type $t\in S'$ when 
  winners are selected with the die rolls $\D'$,
  assuming all types outside $\hat{S}$ always get negative die rolls. 

  None of the monotonicity properties of the functions $g_t(\cdot)$
  will be affected under this new definition;
  thus, the same proof will go through,
  with only a slight change in the limits of $\pmf{k}g_k(\theta_k, \suf_{k+1})$:
  \begin{align*}
    \lim_{\theta_k \to \infty} \pmf{k} g_k(\theta_k, \suf_{k+1})
    & \in \{ \sum_{t \ge k} \pmf{t} \interim{t},
               R(A\cup \SET{k})-\sum_{t\in A} \pmf{t} \interim{t}\}\\
    \lim_{\theta_k\to 0} \pmf{k} g_k(\theta_k, \suf_{k+1})
    & \in \{ 0, \sum_{t \ge k} \pmf{t} \interim{t} -
           \left(R(A) - \sum_{t\in A \setminus B} \pmf{t} \interim{t} \right)
             \}. \QED
  \end{align*}
\end{noqedproof}
}{}

\begin{extraproof}{Theorem~\ref{thm:exist}}
First consider the case when $\M$ is non-separable.
Let $T$ be the type set with $m$ types.
We define the dice system as follows:
First, for all types $t$ with $\interim{t} = 1$,
assign them a point distribution (single-sided die) at $2$, so they
always win.
Next, for all types $t$ with $\interim{t} = 0$,
assign them a point distribution at $-1$, so they never win. 
Next, we create dice for barrier sets $S$ according to
Lemma~\ref{lem:exist-barrier}.
Then, starting from the barrier sets,
we repeatedly apply Lemma~\ref{lem:exist-lattice} to construct dice
for larger tight sets $\hat{S} \supsetneq S$ (or all of $T$)
implementing $\INTE_{\hat{S}}$.

If $\M$ is separable, let $\M_1, \dots, \M_k$ be the components of
$\M$.
Using the construction from the previous paragraph for each $\M_j$,
let $\D_j$ be the dice set constructed for $\M_j$.
Since there is no circuit containing two candidates from
different components,
the winner set of one component has no effect on the winner set of any
other component.
Thus, the union $\bigcup_{j=1}^k \D_j$ of dice implements the desired
interim rule.
\end{extraproof}

\subsection{Winner-Selecting Dice with polynomially many faces}

The proof is based on the following generalization of the fundamental
theorem of linear programming to uncountable dimensions.

\begin{theorem}[Theorem B.11 from \cite{lasserre2010moments}]
\label{thm:lasserre}
Let $f_1, \dots, f_m : \mathbb{X} \to \mathbb{R}$ be Borel measurable
on a measurable space $\mathbb{X}$,
and let $\mu$ be a probability measure on $\mathbb{X}$ such that $f_i$
is integrable with respect to $\mu$ for each $i=1,\dots, m$.
Then, there exists a probability measure $\varphi$ on $\mathbb{X}$
with finite support, such that
\camera{
\[
  \int_{\mathbb{X}} f_id\varphi = \int_{\mathbb{X}} f_id\mu
\]}{%
$
  \int_{\mathbb{X}} f_id\varphi = \int_{\mathbb{X}} f_id\mu
$}
for all $i$.
Moreover, there is such a $\varphi$ whose support consists of at most
$m+1$ points.
\end{theorem}

\begin{extraproof}{Theorem~\ref{thm:finite-face}}
Theorem~\ref{thm:exist} establishes that there is a vector of probability
measures $\D=(D_t)_{t\in T}$ over $V = (1, 2) \cup \SET{-1}$,
satisfying the following for all $i$ and $t_i \in T_i$:
\begin{equation}\label{eq:integral}
 \int_{V} \sum_{\vc{t}\in T_{-i}}\prod_{j\ne i}
    \pmf{t_j} \idotsint_{V^{n-1}} w_{t_i}(s_1, \dots, s_n)
    dD_{t_1}(s_1)\cdots  dD_{t_n}(s_n)dD_{t_i}(s_i)
  = \interim{t_i}.
\end{equation}
where $w_{t_i}(s_1, \dots, s_n)$ equals $1$ if type $t_i$ is a winning type
with dice rolls $s_1, \dots, s_n$, and equals $0$ otherwise. 

For a fixed type $t^\ast\in T_i$ and an equation corresponding to $(j, t'_j)$,
we can change the order of integration to make $dD_{t^\ast}$
the outermost integral,
and isolate the terms that do not involve $dD_{t^\ast}$.
Specifically, we define
  \[
  q_{j, t'_j}(s_i) =
\sum_{\substack{\vc{t}\in T\\ t_j=t'_j , t_i=t^\ast}}\prod_{k\ne j}
    \pmf{t_k} \idotsint_{V^{n-1}} w_{t_j}(s_1, \dots, s_n)
    dD_{t_1}(s_1)\cdots  dD_{t_n}(s_n)
  \]
to be the inner integral over distributions of $t\ne t^\ast$,
as a function of $s_i$, and
\[
  c_{j, t'_j} =
  \begin{cases}
    0, \text{ for } t'_j=t^\ast\\
    \begin{aligned}
      \int_{V} \sum_{\vc{t}\in T_{-j}, t_i\ne t^\ast} \prod_{k\ne j}
      \pmf{t_k} \idotsint_{V^{n-1}} &w_{t_j}(s_1, \dots, s_n)\\
      &dD_{t_1}(s_1)\cdots  dD_{t_n}(s_n)dD_{t_j}(s_j), \text{ for }  t'_j\ne t^\ast
    \end{aligned}
  \end{cases}
\]
as the component of the integral that does not have $dD_{t^*}$ involved.
Thus, Equation~\eqref{eq:integral} can be rewritten as 
\camera{%
\[
  \int_{V} q_{j, t'_j}(s_i) d D_{t^\ast}(s_i) + c_{j, t'_j}
  = \interim{t'_j}.
\]}{%
$
  \int_{V} q_{j, t'_j}(s_i) d D_{t^\ast}(s_i) + c_{j, t'_j}
  = \interim{t'_j}.
$}
By Theorem~\ref{thm:lasserre},
$D_{t^\ast}$ can be changed to a measure $D'_{t^\ast}$ with support of
size $m+1$, such that for all $j$ and $t_j \in T_j$,
\camera{%
\[
  \int_V q_{j, t_j} (s_i) d D_{t^\ast}(s_i)
  = \int_V q_{j, t_j}(s_i)d D'_{t^\ast}(s_i).
\]}{%
$
  \int_V q_{j, t_j} (s_i) d D_{t^\ast}(s_i)
  = \int_V q_{j, t_j}(s_i)d D'_{t^\ast}(s_i).
$}
In other words, we can replace $D_{t^*}$ with $D'_{t^*}$,
which has at most $m+1$ faces. 
Applying the same procedure to each type $t^*$ in turn,
all dice can be replaced by dice with at most $m+1$ faces. 
\end{extraproof}


\section{Efficient Construction of Dice for
  Single-Winner Environments}
\label{sec:interim-to-dice}

In the preceding section, we proved the existence of winner-selecting
dice by a ``construction.''
However, the construction involves repeated appeals to the
Intermediate Value Theorem, and is thus inherently
non-computational.
It is certainly not clear how to implement it efficiently.  
In this section, we show that when the matroid $\M$ is $1$-uniform,
i.e., at most one winner can be selected,
winner-selecting dice can be computed efficiently.
We prove the following. 

\begin{theorem} \label{thm:dice-independent}
Consider a winner-selection environment with \NUMC candidates,
where each candidate $i$'s type is drawn independently from a prior
\PMF[i] supported on $T_i$, and at most one winner can be selected.
If \INTE is a feasible interim rule for
$\PMF = \PMF[1] \times \cdots\times \PMF[n]$,
an explicit representation of the associated dice can be
computed in time polynomial in \NUMC and $\NUMT = \sum_i |T_i|$. 
\end{theorem}

We use the same notation \pmf{t} and \interim{t} as in
Section~\ref{sec:exist}. 
In the single-winner setting, the function
$R(S) = \Expect[\vc{t}\sim \PMF]{r_\M(\vc{t}\cap S)}$
can be treated as a natural extension of \PMF to subsets of $T$:
given  $S \sse T$, we let
$R(S) = \pmf{S} = 1- \prod_{i=1}^{\NUMC} (1- \sum_{t \in S_i} \pmf{t})$,
which is the probability that at least one type in $S$ shows up.
Border's constraints can then be written as follows:
$
  \sum_{t \in S} \pmf{t} \interim{t}
  \; \leq \; \pmf{S} \mbox{ for all } S \subseteq T.
$
The slack function becomes
$\slack_{f,\INTE}(S) = \pmf{S} - \sum_{t \in S} \pmf{t} \interim{t}$
and is nonnegative everywhere when $\INTE$ is feasible.

Tight sets and barrier sets are defined as the special case of the
definition for general matroids in  Section~\ref{sec:exist}. 
The algorithm \FindBarrierSet, given as Algorithm~\ref{alg:findbarrier},
simply implements the definition of barrier sets as minimal non-empty
tight sets, and therefore correctly computes a barrier set.

\begin{algorithm} 
  \DontPrintSemicolon
    $T^+ \gets \set{t \in T: f(t) \interim{t} > 0}$.\;
    Define $g(S) = f(S) - \sum_{t \in S} f(t) \interim{t}$ for $S \sse T^+$.\;
    \If{$\min\set{g(S) : \emptyset \subsetneq S \subseteq T^+} \neq 0$ \label{step:min1}}
    {\Return $T^+$.}
    \Else
    {
      $T^* \gets T^+$.\;
       \textbf{while} there is a type $t \in T^*$ such that $\min\set{g(S): \emptyset \subsetneq S \subseteq T^* \sm \set{t}} = 0$ \label{step:min2}
        \textbf{do} $T^* \gets T^* \sm \set{t}$.\;
      \Return $T^*$.\;}
\caption{$\FindBarrierSet (f,\INTE)$ \label{alg:findbarrier}}
\end{algorithm}

The following lemma characterizes the key useful structure of barrier sets for the single-winner setting. 

\begin{lemma}\label{lem:minimal}
Let $\PMF$ and \INTE be such that \INTE is feasible for $f$.
If there are multiple barrier sets for $(f,\INTE)$,
then there is a candidate $i^*$ such that
each barrier set is a singleton $\set{t}$ with $t \in T_{i^*}$.
\end{lemma}

\camera{In other words, either there is a unique barrier set, or all
barrier sets are singletons of types from a single candidate.}{}

\begin{proof}
Let $A, B$ be any two barrier sets.
Because $A$ and $B$ are both tight,
the lattice property of tight sets implies that $A \cap B = \emptyset$.

We first show that there is a candidate $i^*$
with $A \subseteq T_{i^*}$ and $B \subseteq T_{i^*}$.
Suppose not for contradiction; then, there exist $i \neq j$
and types $t_i \in A \cap T_i$ and $t_j \in B \cap T_j$.
With non-zero probability, the types $t_i$ and $t_j$ show up at the
same time.
However, according to the definition of a tight set,
when a type in $A$ shows up,
the winner must be a candidate with type in $A$,
and the same must hold for $B$.
Then, the winner's type would have to be in $A \intersect B$ with
non-zero probability.
This contradicts the disjointness of $A$ and $B$. 

It remains to show that all barrier sets are singletons.
Since $A$ is tight and all types in $A$ belong to the same candidate,
$\sum_{t \in A} f(t) \interim{t} = f(A) = \sum_{t \in A} f(t)$.
Hence, $\interim{t} = 1$ for all $t \in A$,
and because barrier sets are minimally tight, $A$ must be a singleton.
\end{proof}

\subsection{Description of the Algorithm}

Given a prior $f$ and a feasible interim rule \INTE,
the recursive procedure \ConstructDice,
shown in Algorithm~\ref{alg:dice-construction},
returns a family of dice \DIESET implementing \INTE for $f$.
It operates as follows.
There are two simple base cases:
when no candidate ever wins,
and when a single type of a single candidate always wins.
In the recursive case, the algorithm carefully selects a type
$t^*$ and awards its die the highest-valued face $M'$.
It assigns this new face a probability as large as possible,
subject to still permitting implementation of \INTE.
We choose $t^*$ as a member of a barrier set; this is important in
order to guarantee that the algorithm makes significant progress.

The subroutine \Decrement, shown as Algorithm~\ref{alg:decrement},
essentially conditions both $f$ and \INTE on the face $M'$ \emph{not}
winning.
Specifically, \Decrement computes the conditional type distribution $f'$,
and an interim rule \INTEP, such that if there were a dice
implementation of \INTEP for $f'$,
then adding $M'$ to the die of $t^*$ would yield a set of dice
implementing \INTE for $f$.

\begin{algorithm}[h]
\DontPrintSemicolon
\Input{PDFs $f_1, \ldots, f_{\NUMC}$ supported on disjoint type sets $T_1,\ldots,T_{\NUMC}$.}
\Input{Interim rule \INTE feasible for $f$.}
\Output{Vector of dice $(\DIE[t])_{ t \in \biguplus_{i=1}^{\NUMC} T_i}$.}
Let $T= \biguplus_i T_i$.\;
Let $T^+_i = \set{t \in T_i : f_i(t) \interim[i]{t} > 0}$,
and let $T^+ = \biguplus_{i=1}^{\NUMC} T^+_i$.\;
\If{$T^+=\emptyset$} {
  \textbf{for all} types $t \in T$, let \DIE[t] be a single-sided die with a $-1$ face.\;
}
\ElseIf{there is a type $t^* \in T^+$ with $f(t^*) \interim{t^*} = 1$}
{
Let \DIE[t^*] be a single-sided die with a $+1$ face.\;
\textbf{for all} other types $t \in T \sm \set{t^*}$,
let \DIE[t] be a single-sided die with a $-1$ face.\;
}
\Else{
  Let $T^* = \FindBarrierSet(f,\INTE)$.\; \label{step:T*}
  Let $t^* \in T^*$ be a type chosen arbitrarily.\; \label{step:t*}
  Let $(f',\INTEP)=\Decrement(f, \INTE, t^*, q^*)$,
  for the largest value of $q^* \in [0, f(t^*) \interim{t^*}]$
  such that \INTEP is feasible for $f'$. \label{step:dec}
  \tcc*[r]{Note that $f(t^*) \interim{t^*} < 1$.} 
  Let $(\DIEP[t])_{ t \in T} \gets \ConstructDice(f', \INTEP)$.\;\label{step:recurse}
  Let $M$ be the maximum possible face of any die \DIEP[t], and $M' := \max(M,0)+1$.\;
  Let $\DIE[t] = \DIEP[t]$ for all types $t \neq t^*$.\; \label{step:newdice1}
  Let \DIE[t^*] be the die which rolls $M'$ with probability $\frac{q^*}{f(t^*)}$,
  and \DIEP[t^*] with probability $1-\frac{q^*}{f(t^*)}$.} \label{step:newdice2} 
\Return $(\DIE[t])_{t \in T}$.
\caption{\ConstructDice($f,\INTE$) \label{alg:dice-construction}}
\end{algorithm}

\begin{algorithm}[h]
\DontPrintSemicolon
\tcc*[l]{$q \geq 0$ is the probability allocated to the highest face.
  Because it is a contribution to the \emph{unconditional} winning
  probability $f(t^*) \interim{t^*}$ of type $t^*$,
  and we separated out the case that a single type has unconditional
  winning probability 1, 
  $q$ satisfies $q \leq f(t^*) \interim{t^*} < 1$.}
\textbf{if} $q=f(t^*)$, \textbf{then}
  let $f'(t^*) \gets 0$
  and $\interimp{t^*} \gets 0$ \;
\textbf{else}
  let $f'(t^*) \gets \frac{f(t^*) - q}{1-q}$
  and $\interimp{t^*} \gets \frac{f(t^*) \interim{t^*} - q}{f(t^*) - q}$.\;
Let $i^*$ be such that $t^* \in T_{i^*}$.\;
\textbf{for all} $t \in T_{i^*}, t \ne t^*$,
let $f'(t) \gets\frac{f(t)}{1-q}$ and $\interimp{t} \gets \interim{t}$.\;
\textbf{for all} $t \in T \sm T_{i^*}$,
let $f'(t) \gets f(t)$ and $\interimp{t} \gets\frac{\interim{t}}{1-q}$.\;
\Return $(f',\INTEP).$
\caption{$\Decrement (f, \INTE, t^*, q)$ \label{alg:decrement}}
\end{algorithm}

We now provide a formal analysis of our algorithm.
Theorem~\ref{thm:dice-independent} follows from
Lemmas~\ref{lem:correctness}--\ref{lem:runtime}.
\begin{lemma}\label{lem:correctness}
  If \ConstructDice terminates, it outputs dice implementing
  \INTE for $f$. 
\end{lemma}
\begin{lemma}\label{lem:iterations}
\ConstructDice terminates after at most $\NUMT^2$ recursive calls.
(Recall that $\NUMT = \sum_i |T_i|$.)
\end{lemma}
\begin{lemma}\label{lem:runtime}
  Excluding the recursive call, each invocation of \ConstructDice can
  be implemented in time polynomial in \NUMC and \NUMT.
\end{lemma}
\subsection{Proof of Lemma~\ref{lem:correctness} (Correctness)}

We prove the lemma by induction over the algorithm's calls.
Correctness is obvious for the two base cases:
when $T^+ = \emptyset$ (no type should win),
and when there exists a type $t^*$ with $f(t^*) \interim{t^*}=1$
($t^*$ always shows up and should always win).
For the inductive step, suppose that the recursive call in
step~\ref{step:recurse} returns dice
$\DIESETP = (\DIEP[t])_{ t \in T}$, correctly implementing \INTEP
for $f'$,
and let $\DIESET = (\DIE[t])_{ t \in T}$ be the new dice defined in
steps \ref{step:newdice1} and~\ref{step:newdice2}. 

We analyze the interim winning probability of each type when using the
dice-based winner selection rule given by \DIESET.
For each type $t$, let $\bar{v}_t \sim \DIE[t]$ be a roll of the die for
type $t$, and for each $i$, let $t_i \sim f_i$ be a draw of a type;
all $\bar{v}_t$ and $t_i$ are mutually independent.
In other words, we may assume that the die of \emph{every} type is
rolled (including types that do not show up),
then the type profile is drawn independently.
The winning type is then the type $t_i$ with largest positive
$\bar{v}_{t_i}$;
if all $\bar{v}_{t_i}$ are negative, then no type wins.

Let $t^* \in T_{i^*}$ be as defined in step~\ref{step:t*}.
Let $\E$ be the event that $i^*$ has type $t^*$ and that $\bar{v}_{t^*} = M'$,
and let $\bar{\E}$ be its complement.
By independence of the random choices, the probability of $\E$ is
$f(t^*) \cdot \frac{q^*}{f(t^*)} = q^*$.
Type $t^*$ always wins under the event $\E$.
Conditioned on $\bar{\E}$, each $\bar{v}_t$ (including $\bar{v}_{t^*}$)
is distributed as a draw from \DIEP[t],
the type vector \TP is distributed as a draw from
$f'_1 \cross \cdots \cross f'_{\NUMC}$,
and the $\bar{v}_t$'s and \TP are mutually independent.
By the inductive hypothesis, conditioned on $\bar{\E}$,
each type $t$ wins with probability $f'(t) \interimp{t}$.
Using the definition of $f'$ and \INTEP from the \Decrement subroutine,
the total winning probability for $t^*$ is
$q^* \cdot 1 + (1-q^*) \cdot f'(t^*) \interimp{t^*}
= q^* + (1-q^*) \cdot \frac{f(t^*) \interim{t^*} - q^*}{1-q^*}
= f(t^*) \interim{t^*}$.
For $t \neq t^*$, the total winning probability is
$q^* \cdot 0  + (1-q^*) f'(t) \interimp{t}
= (1-q^*) \cdot \frac{f(t) \interim{t}}{1-q^*}
= f(t) \interim{t}$.
Therefore, the interim winning probability for each type $t$ is \interim{t},
and the dice \DIESET implement \INTE for $f$.

\subsection{Proof of Lemma~\ref{lem:iterations} (Number of Recursive Calls)}

The following lemma is essential in that it shows that invoking
\Decrement maintains feasibility and tightness of sets.

\begin{lemma} \label{lem:decrement}
Let $f$, \INTE, $t^*$ and $q$ be valid inputs for \Decrement,
and $f'$, \INTEP the output of the call to $\Decrement (f, \INTE, t^*, q)$.
Let $S$ be any set of types with $t^* \in S$. Then,
\begin{enumerate}
\item The Border constraint for $S$ is satisfied for $(f', \INTEP)$
  if and only if it is satisfied for $(f, \INTE)$.
\item The Border constraint for $S$ is tight for $(f', \INTEP)$
  if and only if it is tight for $(f,\INTE)$. 
\end{enumerate}
\end{lemma}

\begin{noqedproof}
We will show that the slack for every $S \ni t^*$ satisfies
$\slack_{f',\INTEP}(S) = \frac{\slack_{f,\INTE}(S)}{1-q}$,
which implies both claims.
Let $i^*$ be such that $t^* \in T_{i^*}$.
Using the definitions of $f'$ and \INTEP,
\begin{align*}
\slack_{f',\INTEP}(S)
  & = f'(S) - \sum_{t\in S} f'(t) \interimp{t} \\
  & = 1 - \left( 1 - \frac{f(S_{i^*})-q}{1-q} \right)
             \cdot \prod_{i \neq i^*} (1-f(S_i))
       - \frac{(\sum_{t\in S}f(t)  \interim{t}) - q}{1-q} \\
  & = \frac{1}{1-q} \cdot \left( 1- (1-f(S_{i^*}))
             \cdot \prod_{i \neq i^*} (1-f(S_i))
       - \sum_{t\in S} f(t) \interim{t} \right)\;=\; \frac{\slack_{f,\INTE}(S)}{1-q}.\QED
\end{align*}
\end{noqedproof}

We are now ready to prove Lemma~\ref{lem:iterations}.
We will show that with each recursive invocation of \ConstructDice
(step \ref{step:recurse}), at least one of the following happens:
(1) The number of active types $|T^+|$ decreases;
(2) The size of the barrier set $|T^*|$ decreases.

Notice that the number of active types or the size of a barrier set
never \emph{increase}.
Because the size of the barrier set can only decrease at most \NUMT
times, the number of active types must decrease at least every \NUMT
recursive invocations.
It, too, can decrease at most \NUMT times, implying the claim of the lemma.

Let $T^*$, $t^*$, and $(q^*, f', \INTEP)$ be as chosen in
steps~\ref{step:T*}, \ref{step:t*}, and \ref{step:dec}, respectively.
Let candidate $i^*$ be such that $t^* \in T_{i^*}$.
If $q^*= f(t^*) \interim{t^*}$,
then the type $t^*$ will be inactive in $(f',\INTEP)$,
and there will be one fewer active type in the subsequent invocation
of \ConstructDice (step \ref{step:recurse}).
We distinguish the cases $|T^*|=1$ and $|T^*| > 1$.

If $|T^*|=1$, then $T^*= \set{t^*}$.
By definition of a barrier set, $T^*$ is tight,
implying (for a singleton set) that $\interim{t^*}=1$.
We claim that $q^*$ is set to $f(t^*) = f(t^*) \interim{t^*}$ in
step~\ref{step:dec},
implying that the number of active types decreases.
To prove that $q^* = f(t^*)$, we will show that this choice of $q^*$
is feasible in the invocation of $\Decrement(f, \INTE, t^*, f(t^*))$.
Consider the $\hat{f}, \INTEH$ resulting from such an
invocation of \Decrement.
Lemma~\ref{lem:decrement} implies that the feasibility of each Border
constraint corresponding to a set $S \ni t^*$ is preserved for
$(\hat{f}, \INTEH)$.
For type sets $S$ excluding $t^*$,
\begin{align*}
\slack_{\hat{f}, \INTEH}(S)
  & = \hspace{-5pt}\hat{f}(S) - \sum_{t\in S} \hat{f}(t) \interimh{t} \;=\; 1 - \left(1-\frac{f(S_{i^*})}{1-f(t^*)} \right)
          \cdot \prod_{i \neq i^*} (1-f(S_i))
        - \frac{\sum_{t\in S}f(t) \interim{t} }{1-f(t^*)} \\
  & =  \begin{aligned}[t]
    \frac{1}{1-f(t^*)} \cdot \Biggl( 1 &- (1-f(S_{i^*}) + f(t^*)) \cdot
      \prod_{i \neq i^*} (1-f(S_i))  \\
         &-  \left( \sum_{t\in S}f(t) \interim{t} + f(t^*)\interim{t^*} \right) \Biggr) 
    \end{aligned}\\
  & = \frac{\slack_{f,\INTE}(S\cup\SET{t^*})}{1-f(t^*)},
\end{align*}
which is nonnegative because \INTE is feasible for $f$.
Therefore, step~\ref{step:dec} indeed chooses $q^* = f(t^*)$.

Next, we consider the case $|T^*| > 1$,
and assume that $q^* < f(t^*) \interim{t^*}$
(since otherwise, we are done).
Then, the set of active types $T^+$ is the same for both
$(f,\INTE)$ and $(f',\INTEP)$.

If the instance $(f', \INTEP)$ for the recursive call has multiple
barrier sets, then by Lemma~\ref{lem:minimal},
they are all singletons, and indeed the size of the barrier set in the
next recursive call (which is 1) is strictly smaller than $|T^*|$.
So we assume that $(f', \INTEP)$ has a unique barrier set $T'$.

Because $|T^*| > 1$, Lemma~\ref{lem:minimal} implies that
$T^*$ is the unique barrier set for $(f,\INTE)$.
Therefore, by the definition of barrier sets,
$T^*$ (and hence also $t^*$) is contained in every tight set of active
types for $(f,\INTE)$ (if any).
Because $t^*$ is contained in all tight sets, Lemma~\ref{lem:decrement}
implies that for every $q \in [0, f(t^*) \interim{t^*}]$,
the result of $\Decrement(f,\INTE,t^*,q)$
does not violate any constraints which are already tight for
$(f,\INTE)$, and in fact preserves their tightness.

Because all other constraints have slack,
the optimal $q^*$ is strictly positive.
By assumption, we also have that  $q^* < f(t^*) \interim{t^*}$;
therefore, a non-empty set $S'$ which was not tight for $(f,\INTE)$
must have become tight for
$(f',\INTEP) = \Decrement (f,\INTE,t^*,q^*)$.
By Lemma~\ref{lem:decrement}, this set $S'$ does not include $t^*$.
Because discarding inactive types preserves tightness,  
we may assume without loss of generality that $S' \subseteq T^+$.


We have shown the existence of a non-empty tight set
$S' \subseteq T^+ \sm \set{t^*}$ for $(f',\INTEP)$.
By definition, the barrier set $T'$ is the (unique, in our case)
minimal tight set for $(f',\INTEP)$, so $T' \subseteq S'$. 
We distinguish two cases, based on the possible definitions of
barrier sets:

\begin{itemize}
\item If $T^* = T^+$, then because $S' \subsetneq T^+ = T^*$,
  the barrier set $T'$ is strictly smaller than $T^*$.
\item If $T^*$ is tight for $(f,\INTE)$, then it is also tight for
  $(f',\INTEP)$ by Lemma~\ref{lem:decrement};
  since both $S'$ and $T^*$ are tight for $(f',\INTEP)$,
  we get that $T' \subseteq T^* \cap S' \subseteq T^* \setminus \set{t^*}$
  is strictly smaller than $T^*$.
\end{itemize}

\subsection{Proof of Lemma~\ref{lem:runtime} (Runtime per Call)}

There are only two steps for which polynomial runtime is not
immediate:
the computation of $q^*$ in step \eqref{step:dec} of \ConstructDice,
and finding a non-empty minimizer of a submodular function in
steps \eqref{step:min1} and \eqref{step:min2} of \FindBarrierSet. 
We prove polynomial-time implementability of both steps in the
following lemmas.

\begin{lemma}
In step \ref{step:dec} of \ConstructDice,
$q^*$ can be computed in $\poly(\NUMT)$ time.
\end{lemma}

\begin{proof}
Let $f$, \INTE, and $t^*$ be as in step \ref{step:dec},
and let the candidate $i^*$ be such that $t^* \in T_{i^*}$.
For each $q \in [0,f(t^*) \interim{t^*}] \sse [0,1)$,
let $(f_q, \INTE_q) = \Decrement(f, \INTE, t^*, q)$
be the result of running $\Decrement(f, \INTE, t^*, q)$ with parameter $q$.
Lemma~\ref{lem:decrement} implies that all Border constraints
for $S \ni t^*$ remain feasible for $(f_q, \INTE_q)$.
For type sets $S \not\ni t^*$, we can write the slack in the
corresponding Border constraint as a function of $q$ as follows: 
\begin{align*}
\slack_{f_q,\INTE_q}(S)
  & = f_q(S) - \sum_{t\in S} f_q(t) x_q(t) \\
  & = 1 - \left(1-\frac{f(S_{i^*})}{1-q} \right)
    \cdot \prod_{i \neq i^*} \left(1-f(S_i) \right)
    - \frac{\sum_{t\in S}f(t) \interim{t} }{1-q} \\
  & = \frac{1}{1-q} \cdot \left( 1- (1 - f(S_{i^*}) - q)
    \cdot \prod_{i \neq {i^*}} (1-f(S_{i}))
    - \sum_{t\in S}f(t)\interim{t} - q \right)\\
  & = \frac{1}{1-q} \left( \slack_{f,\INTE}(S) - q f(S \sm S_{i^*})\right).
\end{align*}
The preceding expression is nonnegative if and only if
$q \leq h(S) := \frac{\slack_{f,\INTE}(S)}{f(S \sm S_{i^*})}$.
Therefore, $q^*$ is the minimum of $f(t^*) \interim{t^*}$ and
$\min_{ S \sse T \sm \set{t^*}} h(S)$.
The function $h(S)$ does not appear to be submodular,
and hence efficient minimization is not immediate.
 We utilize Theorem~\ref{thm:prefix} to reduce the search
space and compute $q^*$ efficiently. 

For an interim rule $\INTE: T \to [0,1]$,
we call $S \sse T$ a \emph{level set} of \INTE if there exists
an $\alpha \in [0,1]$ such that $S = \set{ t \in T : \interim{t} > \alpha}$.
If $q^*= \min_{S \not\ni t^*} h(S)$,
then at least one Border constraint just becomes tight at $q=q^*$.
Theorem~\ref{thm:prefix} implies that at least one level set of
$\INTE_{q^*}$ corresponds to one of these newly tightened constraints,
and $h$ is minimized by such a level set.
It follows that, in order to compute $q^*$,
it suffices to minimize $h$ over all those sets $S \not\ni t^*$
which could possibly arise as level sets of some
$\INTE_q$ for $q \in [0, f(t^*) \interim{t^*}]$. 

Let $t_1, \ldots, t_K$ be the types in $T_{i^*} \sm \set{t^*}$,
ordered by non-increasing \interim{t};
similarly, let $t'_1, \ldots, t'_L$ be the types in $T \sm T_{i^*}$,
ordered by non-increasing \interim{t}.
The relative order of types in $T_{i^*} \sm \set{t^*}$ is the same
under $x_q(t)$ as under \interim{t}, because $x_q(t) = \interim{t}$;
similarly, the relative order of types in $T \sm T_{i^*}$ is the same
under $x_q(t)$ as under \interim{t}, because $x_q(t) = \frac{\interim{t}}{1-q}$.
Therefore, the family
$\left\{\set{t_1, \ldots, t_k, t'_1, \ldots, t'_\ell} : k \leq K, \ell \leq L\right\}$
includes all level sets of every $\INTE_q$ excluding $t^*$.
There are at most $\NUMT^2$ type sets in this family,
and those sets can be enumerated efficiently to minimize $h$.
\end{proof}

\begin{lemma}
There is an algorithm for computing a non-empty minimizer of a
submodular function in the value oracle model,
with runtime polynomial in the size of the ground set. 
\end{lemma}
\begin{proof}
For a submodular function $g: 2^T \to \mathbb{R}$,
let $g_t(S) = g (S\cup \SET{t})$, for $t\in T$.
$g_t$ is also submodular,
and can be minimized in time polynomial in $|T|$ \citep{GLSbook}.
Let $S_t$ be a minimizer of $g_t$ and  $t^* \in \arg\min_{t\in T}
g_t(S_t)$. $S_{t^*} \cup \SET{t^*}$ is a non-empty minimizer of $g$. 
\end{proof}


\camera{\section{Symmetric Dice For I.I.D. Candidates in Single-Winner Settings}
\label{sec:iid}
Even when the candidates are identical and the interim rule is symmetric,
Algorithm~\ref{alg:dice-construction} typically produces different
dice for different candidates.
In this section, we design an algorithm specifically for the case of
$\NUMC \geq 2$ i.i.d.~candidates that guarantees symmetry across
candidates in the dice-based implementation. 
\begin{theorem} \label{thm:dice-iid}
Consider a winner-selection environment with \NUMC candidates,
where each candidate has types $T$, each candidate's type is drawn independently from the prior
\PMF on $T$, and at most one candidate is selected as winner.
If $\INTE: T \to [0,1]$ is a feasible symmetric interim rule,
then it admits a symmetric dice-based implementation.
An explicit representation of the associated dice can be
computed in time polynomial in \NUMC and $\NUMT=|T|$. 
\end{theorem}
Algorithm~\ref{alg:dice-construction-iid},
which is similar to Algorithm~\ref{alg:dice-construction},
recursively constructs a dice-based implementation for the
symmetric interim rule \INTE.
\begin{algorithm}[h]
\DontPrintSemicolon
\Input{Number of candidates $\NUMC\geq 2$.}
\Input{PDF $f$ supported on $T$.}
\Input{Symmetric interim rule $\INTE: T \to [0,1]$ feasible for $f^\NUMC$.}
\Output{Set of dice $(\DIE[t])_{ t \in T}$.}
Let $T^+\gets \Set{t}{f(t)\interim{t}> 0}$ be the set of active types.\;
\If{$|T^+|\le 1$}{
   \textbf{for all} $t \in T\backslash T^+$, let \DIE[t] be a
   single-sided die with a $-1$ face.\;
   \textbf{for all} $t \in T^+$ (if any), let \DIE[t] be a
   die with two sides \SET{-1, 1} and probability 
   $\frac{1-\sqrt[\NUMC]{1 - \NUMC f(t) \interim{t}}}{f(t)}$
   of coming up 1.

}
\Else{ 
Let $T^* = \FindBarrierSetIID(\NUMC,f,\INTE)$ \label{findbarrieriid}.\;
Let $t^* \in T^*$ be arbitrary.\;
Let $(f',\INTEP) \gets \DecrementSymmetricDice(\NUMC,f,\INTE,t^*, q^*)$,
for the largest value $q^*\in [0, f(t^*)]$
such that \INTEP is feasible for $f'$. \label{decsymm}\;
$(\DIEP[t])_{t\in T} \gets \ConstructSymmetricDice (\NUMC,\PMFP,\INTEP)$.\;
Let $M$ be the maximum possible face of any \DIEP[t], and $M':=\max(M,0)+1$.\;
Let $\DIE[t] = \DIEP[t]$ for $t\ne t^*$.\;
Let \DIE[t^*] be the die which rolls $M'$ with probability $\frac{q^*}{f(t^*)}$, and rolls \DIEP[t^*] with probability $1-\frac{q^*}{f(t^*)}$.\;}
\Return{$(\DIE[t])_{ t\in T}$}.\;
\caption{$\ConstructSymmetricDice(\NUMC,\PMF,\INTE)$ \label{alg:dice-construction-iid}}
\end{algorithm}
\begin{algorithm}[h]
\DontPrintSemicolon
\tcc*[l]{As in the algorithm \Decrement, $q$ is the probability
  assigned to the highest face.
  Unlike in the algorithm \Decrement, due to symmetry, for any one
  candidate, the probability of winning thanks to this highest face is only
 $(1-(1-q)^\NUMC)/\NUMC $.}
\textbf{if} $\pmf{t^*} = q$, \textbf{then} let $\pmfp{t^*} \gets 0, \interimp{t^*} \gets 0$.\;
\textbf{else} let $\pmfp{t^*} \gets \frac{\pmf{t^*}-q}{1-q}$ and
$\interimp{t^*} \gets \frac{f(t^*)\interim{t^*} -
  (1-(1-q)^\NUMC)/\NUMC}{(f(t^*)-q)(1-q)^{\NUMC-1}}.$
\tcc*[r]{$q < 1$ because $|T^+| \geq 2$.}
\textbf{for all} $t \ne t^*$ \textbf{do}
let $\pmfp{t} \gets \frac{\pmf{t}}{1-q}$
and $\interimp{t} \gets \frac{\interim{t}}{(1-q)^{\NUMC-1}}$.\;
\Return{$(f',\INTEP)$}.\;
\caption{$\DecrementSymmetricDice (\NUMC, \PMF, \INTE, t^*, q)$}
\end{algorithm}

Interim rule feasibility is characterized by the symmetric Border
constraints (Inequality~\eqref{eqn:symmetric-border}).
We let $\slack_{\NUMC,f,\INTE}(S)$ denote the slack in the
symmetric Border constraint for $S \sse T$.
Similarly to Section~\ref{sec:interim-to-dice},
when \INTE is feasible for $(\NUMC,f)$,
we call a set $S \sse T$ \emph{tight} for $(\NUMC,f,\INTE)$
if $\slack_{\NUMC,f,\INTE}(S) = 0$.
Also similar to Section~\ref{sec:interim-to-dice},
we let $T^+= \set{t \in T: f(t) \interim{t} > 0}$ be the set of
\emph{active types},
and we define the \emph{barrier sets} to be the minimal non-empty
tight sets of active types if any non-empty tight sets exist;
otherwise, $T^+$ is the unique barrier set.
The algorithm \FindBarrierSetIID is essentially
identical to \FindBarrierSet from Section~\ref{sec:interim-to-dice};
hence, we omit it here.
The following lemma shows that the barrier set is unique in the
i.i.d.~setting, and that therefore, every tight set contains it.
\begin{lemma}\label{lem:barrierunique}
In the i.i.d.~setting, there is a unique barrier set.
\end{lemma}
\begin{proof}
The uniqueness of the barrier set can be shown by a proof very similar
to the proof of Lemma~\ref{lem:minimal}.
Assume for contradiction that there are two barrier sets $A$ and $B$.
Minimality implies that $A$ and $B$ are disjoint.
Since there are at least 2 candidates, with positive probability,
at least one candidate has a type in $A$ and at least one candidate
has a type in $B$.
In this event, tightness implies that the winner's type must lie in
both $A$ and $B$, contradicting their disjointness.
\end{proof}
We now sketch the proof of Theorem~\ref{thm:dice-iid},
which is very similar to that of Theorem~\ref{thm:dice-independent}.
\begin{extraproof}{Theorem~\ref{thm:dice-iid}}
First, we show correctness.
Assuming that the algorithm terminates,
we show that it outputs a set of dice implementing the symmetric rule
\INTE for $f^\NUMC$.

In the base case $|T^+| \leq 1$, our choice of dice for types 
$t \notin T^+$ ensures that they are never chosen as winner.
If there is a (unique) type $t \in T^+$, then the probability of its
die rolling 1 is
$\frac{1 - \sqrt[\NUMC]{1 - \NUMC f(t) \interim{t}}}{f(t)}$.
The probability that a particular candidate has type $t$ and rolls 1
is therefore $1 - \sqrt[\NUMC]{1 - \NUMC f(t) \interim{t}}$,
and the probability that \emph{at least one} candidate has type $t$
and rolls 1 is $\NUMC f(t) \interim{t}$.
Thus, the probability that \emph{some} candidate has type $t$
and wins is $\NUMC f(t) \interim{t}$;
by symmetry and the uniform tie breaking rule, the probability that
\emph{a specific} candidate has type $t$ and wins is 
$f(t) \interim{t}$.

For the inductive step, consider a recursive call to \ConstructSymmetricDice.
Assume by induction hypothesis that the dice collection
$(\DIEP[t])_{ t \in T}$ implements \INTEP for $(f')^\NUMC$,
and consider the winner selection rule implemented by the dice
$(\DIE[t])_{ t \in T}$ for the joint type distribution $f^\NUMC$.

The probability of at least one candidate rolling the highest face $M'$
is $1 - (1-q^*)^\NUMC$. 
Conditioned on no candidate rolling the highest face $M'$,
the conditional type distribution of each candidate is $f'$,
the conditional distribution of type $t$'s die is \DIEP[t],
and all these distributions are independent.
The probability of a specific candidate winning with type
$t \neq t^*$ is $(1-q^*)^\NUMC f'(t) \interimp{t} = f(t) \interim{t}$.

Next consider the winning probability for type $t^*$.
There are two ways in which a candidate with type $t^*$ could win:
by rolling the highest face $M'$ and winning the uniformly random tie
breaking if necessary,
or by having the highest face when no candidate rolls the face $M'$.
Similar to the argument in the base case,
the probability that \emph{some} candidate wins by rolling $M'$
is $(1-(1-q^*)^\NUMC)$, so the probability that a \emph{specific}
candidate wins by rolling $M'$ is $\frac{1}{\NUMC}(1-(1-q^*)^\NUMC)$.
The probability of winning without rolling the highest face is
$(1-q^*)^\NUMC f'(t^*) \interimp{t^*}$.
In total, the probability of a specific candidate winning with type $t^*$ is
$\frac{1}{\NUMC}(1-(1-q^*)^\NUMC) + (1-q^*)^\NUMC f'(t^*) \interimp{t^*}
= f(t^*) \interim{t^*}$,
so for all types $t$,
the interim winning probability of each type $t$ is \interim{t}, as claimed.

Next, we bound the number of iterations by $\NUMT^2$.
As before, we show that with each recursive call,
either the size of the unique barrier set decreases,
or the number of active types decreases.
The following equation captures the key part of the analogue of
Lemma~\ref{lem:decrement}, 
showing that invoking $\DecrementSymmetricDice(\NUMC,f,\INTE,t^*,q)$
preserves slackness and tightness of every set $S \ni t^*$.
\begin{align*}
\slack_{\NUMC,f', \INTEP}(S)
  & = 1 - \left( 1- \sum_{t \in S} f'(t) \right)^\NUMC
        - \NUMC \sum_{t\in S} f'(t)\interimp{t}\\
  & = 1- \left( 1 - \frac{\sum_{t \in S} f(t)-q}{1-q}\right)^\NUMC
        - \NUMC \cdot \frac{\sum_{t \in S} f(t)\interim{t} - (1-(1-q)^\NUMC)/\NUMC}{(1-q)^\NUMC}\\
  & = \frac{1}{(1-q)^\NUMC} \cdot
    \left( (1-q)^\NUMC - \left( 1 - \sum_{t \in S} f(t) \right)^\NUMC
        - \NUMC \sum_{t\in S} f(t)\interim{t} + (1-(1-q^\NUMC)) \right) \\
  & = \frac{\slack_{\NUMC,f,\INTE}(S)}{(1-q)^\NUMC}.
\end{align*}
Now consider a call to \ConstructSymmetricDice.
There are again two cases,
based on the cardinality of the barrier set $T^*$.
If $|T^*| > 1$, then either $q^*= f(t^*)$,
and the number of active types obviously decreases,
or $q^* < f(t^*)$,
and a new tight set is created in step~\eqref{decsymm}.
In the latter case, the size of the barrier set decreases in the next
recursive call;
as in the proof of Lemma~\ref{lem:iterations},
this follows from the uniqueness of the barrier set
(Lemma~\ref{lem:barrierunique}),
the fact that \DecrementSymmetricDice preserves tightness and
slackness of sets containing the type $t^*$
(including, in particular, the tight set $T^*$),
and the fact that tight sets are closed under intersection. 

If $T^*=\set{t^*}$ is a singleton, then tightness implies that
$\NUMC f(t^*) \interim{t^*} = 1-(1-f(t^*))^\NUMC$.
The following calculation then shows that invoking
$\DecrementSymmetricDice(\NUMC,f,\INTE,t^*,q)$ with $q = f(t^*)$
does not violate the feasibility of any set $S \not\ni t^*$. 
\begin{align*}
\slack_{\NUMC,f', \INTEP}(S)
  & = 1 - \left( 1 - \sum_{t \in S} f'(t) \right)^\NUMC
        - \NUMC \sum_{t\in S} f'(t)\interimp{t}\\
  & = 1 - \left( 1 -\frac{\sum_{t \in S} f(t)}{1-q}\right)^\NUMC
        - \NUMC\sum_{t\in S} \frac{f(t)\interim{t}}{(1-q)^\NUMC}\\
  & = \frac{1}{(1-q)^\NUMC} \cdot
    \left( (1-q)^\NUMC - \left( 1-q-\sum_{t \in S} f(t) \right)^\NUMC
                      - \NUMC\sum_{t\in S} f(t)\interim{t} \right) \\
  & = \frac{1}{(1-q)^\NUMC}
    \left( 1 - \left( 1-f(t^*)-\sum_{t \in S} f(t) \right)^\NUMC
    - \NUMC \sum_{t\in S} f(t)\interim{t}
    - \left(1- \left( 1-f(t^*) \right)^\NUMC \right) \right) \\
  & = \frac{\slack_{\NUMC,f,\INTE}(S\cup\SET{t^*})}{(1-q)^\NUMC}
  \; \ge \; 0.
\end{align*}
Because feasibility is also preserved for sets $S \ni t^*$ (as
mentioned above),
this implies that $q^*=f(t^*)$;
as a result, the number of active types decreases.

Finally, it remains to show that each recursive call can be
implemented in time polynomial in \NUMC and \NUMT.
For step~\eqref{findbarrieriid},
the proof is essentially identical to the analogous claim in
Section~\ref{sec:interim-to-dice} and is therefore omitted.
For computing $q^*$ in Step~\eqref{decsymm},
the proof is similar to the analogous claim in Section~\ref{sec:interim-to-dice}.
Let $(f_q,\INTE_q) = \DecrementSymmetricDice(\NUMC,f,\INTE,t^*,q)$;
as shown above, this operation preserves the feasibility of every
constraint corresponding to  $S \ni t^*$.
For sets $S \not\ni t^*$, there exists a threshold $h(S)$ for the
maximum value of $q \leq f(t^*) \interim{t^*}$ such that
$\slack_{\NUMC,f_q,\INTE_q}(S) \geq 0$,
and this threshold can be computed numerically.
As in the proof of Lemma~\ref{lem:runtime},
Theorem~\ref{thm:prefix} implies that it suffices to restrict
attention to \emph{level sets} of the form
$S(\alpha) = \set{ t  \neq t^* : x_q(t) > \alpha}$
for some $q \in [0, f(t^*)]$ and $\alpha \in [0,1]$.

In the i.i.d.~setting, there are at most $\NUMT-1$ such level sets,
since the relative order of types $t \neq t^*$ by  $x_q(t)$ does not
depend on $q$ (and hence, can be computed from \INTE).
As in the proof of Lemma~\ref{lem:runtime},
computing $q^*$ then reduces to minimizing $h(S)$
over these $\NUMT-1$ level sets.
\end{extraproof}
}{}

\section{From Winner-Selecting Dice to Interim Rules for Single-Winner Settings}
\label{sec:dice-to-interim}

Having shown how to compute winner-selecting dice implementing a given interim rule, 
we next show the easier converse direction:
how to compute the interim rule given winner-selecting dice
$\SET{\die[i]{t}}$ in single-winner environments.
As before, we denote the type set of candidate $i$ by $T_i$, and
assume without loss of generality that the type sets of different
candidates are disjoint.
For simplicity of exposition,
we assume that each die has a given finite support%
\footnote{Our approach extends easily to the case of continuously
  supported dice, so long as we can perform integration with respect
  to the  distributions of the various dice.},
and we write $U_i := \bigcup_{t \in T_i} \supp(\die[i]{t})$
for the combined support of candidate $i$'s dice.
We also assume that we can evaluate the probability
$\Prob{\die[i]{t} = u}$ of the face labeled with $u$,
for all candidates $i$, types $t$, and faces $u \in U_i$.

\begin{theorem} \label{thm:dice-to-interim}
Consider a single-winner selection environment with \NUMC candidates
and independent priors $\PMF[1], \ldots, \PMF[\NUMC]$,
where \PMF[i] is supported on $T_i$.
Given dice $\Set{\die[i]{t}}{i \in [\NUMC], t \in T_i}$
represented explicitly,
the interim rule of the corresponding dice-based winner selection rule
can be computed in time polynomial in \NUMC, $\NUMT = \sum_i |T_i|$,
and the total support size $\SetCard{\bigcup_i U_i}$ of all the dice. 
\end{theorem}

\begin{proof}
First, we can compute the probability mass function of each candidate
$i$'s (random) score \score{i}.
Given $u \in U_i$, we have $\Prob{\score{i} = u} = \sum_{t \in T_i} \pmf[i]{t} \cdot \Prob{\die[i]{t} = u}.$

From this probability mass function, we easily compute
$\Prob{\score{i} \leq u}$
for each $u \in U_i$ by the appropriate summation. 
When all dice faces are distinct, this is all we need;
since $\Prob{\score{i'} < u} = \Prob{\score{i'} \leq u}$
for $i' \neq i$ and $u \in U_i$,
the interim rule is given by the following simple equation:
\[
  \interim[i]{t}
   = \sum_{u \in \supp(\die[i]{t})} \Prob{\die[i]{t} = u}
   \cdot \prod_{i' \neq i} \Prob{\score{i'} \leq u}.
\]

When the dice's faces are not distinct,
recall that we break ties uniformly at random.
To account for the contribution of this tie-breaking rule,
we need the distribution of the number of other candidates that tie
candidate $i$'s score of $u$;
this is a Poisson Binomial distribution.
More, precisely, we need the Poisson Binomial distribution
with the $\NUMC-1$ parameters
$\left(\frac{\Prob{\score{i'} = u}}{\Prob{\score{i'} \leq u}} \right)_{i' \neq i}$;
we denote its probability mass function by $B_{i,u}$.
It is well known, and easy to verify, that a simple dynamic program
computes the probability mass function of a Poisson Binomial
distribution in time polynomial in its number of parameters.
Therefore, we can compute $B_{i,u}(k)$ for each
$i \in [\NUMC], u \in U_i$
and $k \in \SET{1, \ldots, \NUMC-1}$.
The interim rule is then given by the following equation:
\[
  \interim[i]{t}
  = \sum_{u \in \supp(\die[i]{t})} \Prob{\die[i]{t} = u} \cdot
    \left( \prod_{i' \neq i}\Prob{\score{i'} \leq u} \right)
       \cdot \sum_{k=0}^{\NUMC-1} \frac{B_{i,u}(k)}{k+1}.
\]

It is easy to verify that all the above computations satisfy the
claimed runtime.
\end{proof}

\section{Winner-Selecting Dice and Persuasion}
\label{sec:persuasion}

In this section, we investigate the existence of winner-selecting dice
for instances of Bayesian persuasion.
Our main result (Theorem~\ref{thm:nodice}) is to exhibit an instance
with independent non-identical actions for which there is no optimal
signaling scheme that can be implemented using winner-selecting dice.
This result is contrasted with Theorem~\ref{thm:iid-persuasion-dice},
which shows that when the actions' types are not just independent,
but identically distributed as well,
a dice-based implementation always \emph{does} exist.

\begin{theorem} \label{thm:nodice}
There is an instance of Bayesian persuasion (given in
Table~\ref{table:eg}) with independent actions
which does not admit a dice-based implementation of any
optimal signaling scheme.
Consequently, there exists a second-order interim rule which does not
admit a dice-based implementation.
\end{theorem}

\begin{theorem} \label{thm:iid-persuasion-dice}
Every Bayesian persuasion instance with i.i.d.~actions admits an
optimal dice-based signaling scheme.
Moreover, when the prior type distribution is given explicitly,
the corresponding dice can be computed in time polynomial in the
number of actions and types.
\end{theorem}

The negative result of Theorem~\ref{thm:nodice} has
interesting implications.
Since second-order interim rules summarize all the attributes of a
winner selection rule relevant to persuasion,
second-order interim rules, unlike their first-order brethren,
can in general not be implemented by dice.
Most importantly, this result draws a sharp contrast between
persuasion and single-item auctions,
despite their superficial similarity:
it rules out a Myerson-like virtual-value characterization of optimal
persuasion schemes, and it joins the \#P-hardness result of
\cite{dughmi2016algorithmic} as
evidence of the intractability of optimal persuasion. 

\subsection{Basics of Bayesian Persuasion}
\label{prelim:persuasion}

In \emph{Bayesian persuasion},
the \NUMC candidates are \emph{actions} which a \emph{receiver} can
take.
Each action $i$ has a type \type{i}, drawn \emph{independently}%
\footnote{The draws are independent in this paper. In more general
  Bayesian persuasion models, they can be correlated.}
from the set $T_i$, according to a commonly known distribution \PMF[i].
Each type $t_i$ has associated payoffs
$s (i, \type{i})$ and $r (i, \type{i})$ for the sender and receiver,
respectively.
The \emph{sender} (or \emph{principal}) also has access to the
\emph{actual} draws $\TP = (\type{1}, \ldots, \type{\NUMC})$
of the types, and would like to use this leverage to persuade the
receiver to take an action favorable to him%
\footnote{To avoid ambiguities, we always
  use male pronouns for the sender and female ones for the receiver.}.

Thereto, the sender can commit to a (typically randomized) policy
\ALGO --- called a \emph{signaling scheme} ---
of revealing some of this information to the receiver.
It was shown by \citet{Kamenica11} that the sender can restrict
attention, without loss, to \emph{direct schemes}:
randomized functions \ALGO mapping type profiles to recommended
actions.
Naturally, the function must be \emph{persuasive}:
if action $i$ is recommended,
the receiver's posterior expected utility from action $i$ must be no
less than her posterior expected utility from any other action $i'$. 
In this sense, direct schemes can be viewed as winner selection rules
in which the actions are the candidates,
and persuasiveness constraints must be obeyed.

\subsection{Proof of Theorems}
\begin{extraproof}{Theorem~\ref{thm:nodice}}
The persuasion instance, shown in Table \ref{table:eg},
features three actions $\SET{A,B,C}$, each of which has two types $\SET{1,2}$.
The types of the different actions are distributed independently.
In the instance, the sender's utility from any particular action is a
constant, independent of the action's type. 
\begin{table}[h!]
  \centering
  \begin{tabular}{c|c|c}
    \diagbox{action}{type}& 1 & 2\\
    \hline
    $A$  & $0.5\times (100, 2)$ & $0.5\times (100, -\infty)$  \\
    \hline    
    $B$  & $0.99\times (1,3)$ & $0.01\times(1, -\infty)$ \\
    \hline
    $C$ & $0.5\times (0,0)$ & $0.5\times (0, 6)$ 
  \end{tabular}
  \caption{A Persuasion instance with no dice-based implementation.
    The notation $p \times (s, r)$ denotes that the type $(s,r)$
    (in which the sender and receiver payoffs are $s$ and $r$, respectively)
    has probability $p$.
    \label{table:eg}}
\end{table}

One (optimal, as we will show implicitly)
signaling scheme is the following.
(In writing a type vector, here and below,
we use $*$ to denote that the type of an action is irrelevant.)
\begin{itemize}[itemsep=0em, partopsep=0em,topsep=0.2em]
\item If the type vector is $(1,*,1)$, then recommend action $A$.
\item If the type vector is $(1,*,2)$, then recommend each of $A, C$
  with equal probability \half.
\item If the type vector is $(2,1,*)$, then recommend action $B$.
\item If the type vector is $(2,2,*)$, then recommend action $C$.
\end{itemize}
While this is not the unique optimal scheme,
we next prove that none of the optimal persuasion schemes admit a
dice-based implementation.

The given signaling scheme recommends action $A$
with probability $3/8$ overall,
action $B$ with probability $\frac{99}{200}$ overall,
and action $C$ with the remaining probability.
No persuasive signaling scheme can recommend $A$ with
probability strictly more than $3/8$, because conditioned on
receiving the recommendation $A$, action $C$ must be at least twice as
likely to be of type $1$ as of type $2$, in addition to action $A$ being of type $1$ with probability $1$.
Similarly, no persuasive scheme can recommend $B$ with probability
strictly more than $\frac{99}{200}$, because action $C$ must be at
least as likely to be of type $1$ as of type $2$ when $C$ is
recommended, in addition to action $B$ being of type $1$ with probability $1$.
Hence, any optimal signaling scheme must recommend $A$ with
probability $3/8$ and $B$ with probability 
$\frac{99}{200}$, and the given scheme is in fact optimal.

Suppose for a contradiction that there exist dice
$(\die[i]{j})_{i \in \SET{A,B,C}, j \in \SET{1,2}}$
implementing an optimal signaling scheme.
We gradually derive properties of these optimal signaling schemes,
eventually leading to a contradiction.

\begin{enumerate}[itemsep=0em]
\item Since action $A$ can never be recommended when it has type 2
  (the receiver would never follow the recommendation),
  it must be recommended with probability $\frac{3}{4}$
  conditioned on having type 1.
\item In particular, whenever the type profile is $(1,*,1)$, action
  $A$ must be recommended, regardless of the type of action $B$.
  This is because action $C$ must be at least twice as likely
    of type 1 as of type 2 for a recommendation of $A$ to be
    persuasive.
\item Therefore, all faces on \die[A]{1} must be larger than all faces
  on \die[B]{1} and on \die[B]{2}.
\item Because of this, action $B$ can never be recommended when the
  type profile is $(1,*,2)$. 
\item Thus, when the type profile is $(1,*,2)$, the signaling scheme
  has to recommend each of $A$ and $C$ with probability \half.
  (The recommendation could of course follow different distributions
  based on the type of $B$; such a correlation is immaterial for our
  argument.)
  \label{step:C-over-B}
\item Given that action $B$ cannot be recommended when action $A$ has
  type 1, or when action $B$ has type 2, it must \emph{always} be
  recommended for type vectors $(2,1,*)$.
\item This implies that all faces of \die[B]{1} must be larger than
  all faces on \die[C]{1} and on \die[C]{2}.
\item This is a contradiction to Step~\ref{step:C-over-B},
  which states that with positive probability,
  \die[C]{2} beats \die[B]{1}.
\end{enumerate}

Thus, we have proved that there is no dice-based implementation of any
optimal signaling scheme for the given instance.
\end{extraproof}

\begin{extraproof}{Theorem~\ref{thm:iid-persuasion-dice}}
When the actions' (or more generally: candidates')
type distributions are i.i.d., i.e.,
$T_i$ and \PMF[i] are the same for all candidates $i$,
\citet{dughmi2016algorithmic} have shown that there is an optimal
\emph{symmetric} signaling scheme,
or more generally a symmetric second-order interim rule \SINTE.
We show that any symmetric second-order interim rule \SINTE is
uniquely determined by its first-order component,
a fact implicit in \citep{dughmi2016algorithmic}.

For symmetric rules, \sinte{i}{i'}{t} depends only on whether $i=i'$
or $i \neq i'$,
but not on the identities of the candidates $i$ and $i'$.
Therefore, \SINTE can be equivalently described by two type-indexed vectors
$\vc{y}$ and $\vc{z}$,
where $y_t= \sinte{i}{i}{t}$ for all candidates $i$,
and $z_t= \sinte{i}{i'}{t}$ for all candidates $i$ and $i'$ with $i \neq i'$.
The vector $\vc{y}$ is a first-order interim rule,
and we refer to it as the \emph{first-order component} of \SINTE.
If \SINTE is feasible and implemented by \ALGO,
then $y_t = \sinte{i}{i}{t} = \ProbC{\ALGO(\TP) = i}{\type{i} = t}$
for all candidates $i$,
so $\vc{y}$ is the first-order interim rule implemented by \ALGO.
For every candidate $i$ and type $t$, we have
\[
  1 \; = \; \sum_{i'=1}^\NUMC \ProbC{\ALGO(\TP) = i'}{\type{i} = t}
  \; = \; (\NUMC-1) z_t + y_t.
\]

Therefore, $\vc{z} = \frac{\vc{1}-\vc{y}}{\NUMC-1}$,
and the first-order component of a symmetric second-order interim rule
suffices to fully describe it.
The second-order rule is also, by the preceding argument,
efficiently computable from its first-order component,
and is feasible if and only if its first-order component is a feasible
symmetric interim rule.
Moreover, by \cite{cai2012algorithmic}, 
feasibility of symmetric second-order interim rules can be checked in
time polynomial in the number of types and candidates,
and given a feasible symmetric second-order interim rule \SINTE,
a winner selection rule implementing \SINTE can be evaluated in
\camera{time polynomial in the number of types and candidates.}{%
  polynomial time.}
\end{extraproof}


\section{Directions for Future Work} \label{sec:conclusions}

We have begun an investigation of dice-based winner selection rules,
in which each of several candidates independently draws a ``score''
from a distribution (rolling a ``die''),
and a candidate set is selected to maximize the sum of scores,
  subject to a feasibility constraint.
We have shown that dice-based winner selection rules can implement all
first-order interim rules with matroid constraints,
but not all second-order interim rules;
in particular, there are instances of Bayesian persuasion in which no
optimal signaling scheme can be implemented using dice.


A natural direction for future work is to understand the limits of
dice-based winner selection rules.
While our existence proof uses matroid properties,
matroid constraints are not the limit of implementability by dice:
in \camera{the appendix}{the full version}, 
we show an example in which the feasible sets do not form a matroid,
yet every feasible interim rule within the environment is
implementable with dice.
This rules out a characterization of the form
``a feasibility constraint $\I$ has all feasible $(\INTE, \PMF)$
implementable by dice if and only if $\I$ is a matroid.''
In fact, we do not know of any feasibility constraint and
corresponding first-order feasible interim rule for which a dice-based
implementation can be ruled out,
though we strongly suspect that such examples exist.
A difficulty in verifying our conjecture is that we are not
aware of a useful general technique for proving the
\emph{non-existence} of a dice-based implementation for a given
interim rule.

Another direction is to find an efficient algorithm for matroid
environments.
To derive an efficient algorithm from our existential proof,
the functions $g_t$ and $h_t$ would have to be evaluated efficiently. 
Furthermore, even if a set of continuous dice is given, 
it is still unclear how to convert them to dice with finitely many
faces in polynomial time.





\bibliographystyle{plainnat}
\bibliography{names,conferences,ref}
\newpage
\camera{%
\appendix
\section{An Interim Rule not Corresponding to any Optimal Auction}
According to \cite{myerson, border}, the linear program
which solves for the revenue-maximizing interim rule
is equivalent to the following LP:

\begin{LP}{Maximize}{%
\sum_{i=1}^\NUMC \sum_{t\in T_i} \pmf[i]{t} \interim[i]{t} \bar{v}_i(t)
}
\INTE \text{ is a feasible interim rule.}
\end{LP}

Here $\bar{v}_i$ is the (ironed) virtual value function.
When no type has virtual value 0, every optimal solution is at the
boundary of the polytope of feasible interim rules.
When some types have virtual value 0,
and we employ the convention of never awarding those types an item,
again we get an interim rule at the boundary.
Therefore, any interior point of this polytope cannot
arise as the interim rule of any optimal auction.
In particular, we have shown that there exist feasible interim
  rules not arising as allocation rules of optimal auctions.



\section{A Non-Matroid Example Implementable by Dice}
\label{sec:non-matroid}
In this section, we show a non-matroid example for which \emph{all}
feasible first-order interim rules are implementable by dice.
Notice that finding \emph{one} interim rule that is implementable by
dice within a given environment is easy:
one can calculate the interim winning probability of a collection of
trivial winner-selecting dice.
What we want is to show that dice exist for all feasible interim
rules of a non-matroid environment.  

Consider a winner-selection environment with $4$ candidates, numbered
$1, 2, 3, 4$.
The feasible winner sets are $\SET{1, 2}, \SET{3,4}$ and
their subsets, i.e., all singletons and the empty set.
Each candidate has a single type, so the type distribution is trivial.
We denote the unique type of candidate $i$ by $i$. 

Consider any ex-post rule, i.e., a distribution over the feasible sets
$\SET{1, 2}, \SET{3, 4}, \SET{1}, \SET{2}, \SET{3}, \SET{4}, \emptyset$.
Denote the ex-post rule's probability of selecting the winner set $S$ by $y_S$.
We have 
\begin{equation}\label{eq:expost-interim}
\begin{aligned}
  x(1) &= y_{\SET{1, 2}} + y_{\SET{1}}\\
  x(2) &= y_{\SET{1, 2}} + y_{\SET{2}}\\
  x(3) &= y_{\SET{3, 4}} + y_{\SET{3}}\\
  x(4) &=y_{\SET{3, 4}} + y_{\SET{4}}
\end{aligned}
\end{equation}
Notice that any feasible interim rule can be implemented by a
distribution $\vc{y}$ in this way.

Without loss of generality, assume that $x(1) \le x(2)$ and $x(3)\le x(4)$. 
Whenever \INTE is feasible, we define the following set of dice.
Each die has $-\infty$ on one side; the other side's value and
probability are given by Table~\ref{tab:face}. 

\begin{table}[ht]
  \centering
  \begin{tabular}{|c|c|c|c|c|}
  \hline
  Type $t$ & 1 & 2 & 3 & 4\\
  \hline
  positive face & 100 & 2 & 10 & 1\\
  \hline
  \makecell{probability of\\ positive face} & $x(1)$ & $\frac{x(2)-x(1)}{1-x(1) -x(3)}$ & $\frac{x(3)}{1-x(1)}$ & $\frac{x(4)-x(3)}{1-x(2)-x(3)}$\\
  \hline
\end{tabular}
\caption{Dice for \INTE \label{tab:face}}
\end{table}

One can easily verify that these dice implement the desired interim rule,
if the given probabilities for the positive face fall in the interval $[0, 1]$. 
Thus, we only need to verify that they do.
Substituting \eqref{eq:expost-interim} into the probabilities from
Table~\ref{tab:face}, we have 

\begin{table}[ht]
\begin{tabular}{|c|c|c|c|c|}
  \hline
  Type $t$ & 1 & 2 & 3 & 4\\
  \hline
  \makecell{probability of \\positive face}  & $y_{\SET{1, 2}} + y_{\SET{1}}$ & $\frac{y_{\SET{2}}-y_{\SET{1}}}{1-y_{\SET{1, 2}} - y_{\SET{1}} -y_{\SET{3, 4}}-y_{\SET{3}}}$ & $\frac{y_{\SET{3, 4}}+y_{\SET{3}}}{1-y_{\SET{1,2}}-y_{\SET{1}}}$ & $\frac{y_{\SET{4}}-y_{\SET{3}}}{1-y_{\SET{1,2}}-y_{\SET{2}}-y_{\SET{3, 4}}-y_{\SET{3}}}$\\
  \hline
\end{tabular}
\end{table}

Since $\sum_{S \subseteq \SET{1,2,3,4}} y_S =  1$ and
$x(1)\le x(2), x(3)\le x(4)$, all the probabilities indeed fall in the
interval $[0, 1]$.
Thus for any feasible interim rule \INTE in this winner-selection environment,
there is a collection of dice implementing \INTE. 

}{}
\end{document}